\documentclass[acmsmall,screen,nonacm]{acmart}
\usepackage[utf8]{inputenc}
\usepackage[ruled]{algorithm}
\usepackage[noend]{algpseudocode}
\usepackage{amsmath}
\usepackage{amsfonts}
\usepackage{mathtools}
\usepackage{todonotes}
\usepackage[ligature,reserved]{semantic}
\usepackage{xspace}

\usepackage{scalerel}
\usepackage{tikz}
\usepackage{adjustbox}

\usetikzlibrary{svg.path}

\definecolor{orcidlogocol}{HTML}{A6CE39}
\tikzset{
    orcidlogo/.pic={
        \fill[orcidlogocol] svg{M256,128c0,70.7-57.3,128-128,128C57.3,256,0,198.7,0,128C0,57.3,57.3,0,128,0C198.7,0,256,57.3,256,128z};
        \fill[white] svg{M86.3,186.2H70.9V79.1h15.4v48.4V186.2z}
        svg{M108.9,79.1h41.6c39.6,0,57,28.3,57,53.6c0,27.5-21.5,53.6-56.8,53.6h-41.8V79.1z M124.3,172.4h24.5c34.9,0,42.9-26.5,42.9-39.7c0-21.5-13.7-39.7-43.7-39.7h-23.7V172.4z}
        svg{M88.7,56.8c0,5.5-4.5,10.1-10.1,10.1c-5.6,0-10.1-4.6-10.1-10.1c0-5.6,4.5-10.1,10.1-10.1C84.2,46.7,88.7,51.3,88.7,56.8z};
    }
}
\newcommand\orcidicon[1]{\href{https://orcid.org/#1}{\mbox{\scalerel*{
                \begin{tikzpicture}[yscale=-1,transform shape]
                \pic{orcidlogo};
                \end{tikzpicture}
            }{|}}}}

\usepackage{hyperref}

\algblockdefx[Upon]{Upon}{EndUpon}%
[1]{{\bf upon} (#1) {\bf do}}%
{{\bf end upon}}

\algblockdefx[When]{When}{EndWhen}%
[1]{{\bf when} (#1) {\bf do}}%
{{\bf end when}}

\makeatletter
\ifthenelse{\equal{\ALG@noend}{t}}%
  {\algtext*{EndUpon}}
  {}%
\makeatother

\makeatletter
\ifthenelse{\equal{\ALG@noend}{t}}%
  {\algtext*{EndWhen}}
  {}%
\makeatother

\mathlig{<-}{\leftarrow}
\mathlig{==}{\equiv}
\mathlig{++}{\mathop{+\!\!+}}
\mathlig{<<}{\langle}
\mathlig{>>}{\rangle}
\mathlig{~}{\sim}
\reservestyle{\variables}{\text}
\variables{epoch,current,history,theset[the\_set],proposal,valid}

\reservestyle{\setops}{\text}
\setops{add,get,Init,Add,Get,BAdd,EpochInc,Broadcast,Deliver,GetEpoch,Propose,Inform,SetDeliver, SuperSetDeliver, DoStuff[Start],Consensus}

\reservestyle{\structs}{\text}
\structs{DPO,BAB,BRB,SBC,DSO}

\reservestyle{\stmt}{\textbf}
\stmt{call,ack,drop,return,assert,wait}

\reservestyle{\mathfunc}{\mathsf}
\mathfunc{send,receive,sendbrb[send\_brb],sendsbc[send\_sbc], count}

\reservestyle{\messages}{\texttt}
\messages{madd[add],mepochinc[epinc]}

\reservestyle{\api}{\texttt}
\api{apiBAdd,apiAdd,apiGet,apiEpochInc,apiTheSet[{the\_set}],apiHistory}

\reservestyle{\schain}{\texttt}
\schain{history,epoch,theset[{the\_set}],add,get,epochinc[{epoch\_inc}],getepoch[{get\_epoch}],tobroadcast[{to\_broadcast}], knowledge, pending, sent, received,proposal,Valid, nop, server}

\reservestyle{\setvars}{\textit}
\setvars{prop,propset}

\newtheorem{definition}{Definition}
\newtheorem{property}{Property}

\newtheorem{lemma}{Lemma}

\usepackage{xspace}
\usepackage{url}
\usepackage{paralist}

\newcommand{\PR}[1]{\ensuremath{\textit{#1}}}
\newcommand{\GS}{\PR{GS}\xspace}
\newcommand{\EPOCH}{\<epoch>}

\newcommand{\HISTORY}{\<history>}
\newcommand{\THESET}{\<theset>}
\newcommand{\epochs}{Setchain\xspace}

\newcommand{\BRBdeliver}{\<BRB>.\<Deliver>}
\newcommand{\BRBbroadcast}{\<BRB>.\<Broadcast>}
\newcommand{\BABdeliver}{\<BAB>.\<Deliver>}
\newcommand{\BABbroadcast}{\<BAB>.\<Broadcast>}
\newcommand{\setchain}{Setchain\xspace}
\newcommand{\setchains}{Setchains\xspace}

\newcommand{\Hspeed}{H1}
\newcommand{\Hbetteralg}{H2}
\newcommand{\Hagr}{H3}
\newcommand{\Hbyzantine}{H4}
\newcommand{\Hdegrade}{H5}

\newcommand{\Propset}{\PR{propset}}
\newcommand{\Servers}{N}

\newcommand{\byzantineset}{\mathbb{B}}
\newcommand{\correctset}{\mathbb{C}}

\newcommand{\AlgTwo}{Alg.~\ref{DPO-alg-basb}\xspace}
\newcommand{\AlgThree}{Alg.~\ref{DPO-alg-basb3}\xspace}
\newcommand{\AlgTwoSet}{Alg.~\ref{DPO-alg-basb}+set\xspace}
\newcommand{\AlgThreeSet}{Alg.~\ref{DPO-alg3-aggregated}\xspace}

\newcommand{\TXONE}[1]{\ensuremath{\textit{#1}}\xspace}
\newcommand{\TXTWO}[2]{\ensuremath{\textit{#1}\_\textit{#2}}\xspace}
\newcommand{\TXTHREE}[3]{\ensuremath{\textit{#1}\_\textit{#2}\_\textit{#3}}\xspace}

\newcommand{\GenerateInvalidElems}{\TXTHREE{generate}{invalid}{elems}}
\newcommand{\HavocNumber}{\TXTWO{havoc}{number}}
\newcommand{\HavocSubset}{\TXTWO{havoc}{subset}}
\newcommand{\HavocPartition}{\TXTWO{havoc}{partition}}
\newcommand{\HavocElement}{\TXTWO{havoc}{element}}
\newcommand{\Knowledge}{\TXONE{knowledge}}

\newcommand{\ValidElements}{\TXTWO{valid}{elements}}
\newcommand{\EV}{\TXONE{ev}}

\newcommand{\Tau}{\mathcal{T}}

\usepackage{listings, xcolor}

\definecolor{verylightgray}{rgb}{.97,.97,.97}

\lstdefinelanguage{Solidity}{
	keywords=[1]{anonymous, assembly, assert, balance, break, call, callcode, case, catch, class, constant, continue, constructor, contract, debugger, default, delegatecall, delete, do, else, emit, event, experimental, export, external, false, finally, for, function, gas, if, implements, import, in, indexed, instanceof, interface, internal, is, length, library, log0, log1, log2, log3, log4, memory, modifier, new, payable, pragma, private, protected, public, pure, push, require, return, returns, revert, selfdestruct, send, solidity, storage, struct, suicide, super, switch, then, this, throw, transfer, true, try, typeof, using, value, view, while, with, addmod, ecrecover, keccak256, mulmod, ripemd160, sha256, sha3}, 
	keywordstyle=[1]\color{blue}\bfseries,
	keywords=[2]{set,address, bool, byte, bytes, bytes1, bytes2, bytes3, bytes4, bytes5, bytes6, bytes7, bytes8, bytes9, bytes10, bytes11, bytes12, bytes13, bytes14, bytes15, bytes16, bytes17, bytes18, bytes19, bytes20, bytes21, bytes22, bytes23, bytes24, bytes25, bytes26, bytes27, bytes28, bytes29, bytes30, bytes31, bytes32, enum, int, int8, int16, int24, int32, int40, int48, int56, int64, int72, int80, int88, int96, int104, int112, int120, int128, int136, int144, int152, int160, int168, int176, int184, int192, int200, int208, int216, int224, int232, int240, int248, int256, mapping, string, elem, uint, uint8, uint16, uint24, uint32, uint40, uint48, uint56, uint64, uint72, uint80, uint88, uint96, uint104, uint112, uint120, uint128, uint136, uint144, uint152, uint160, uint168, uint176, uint184, uint192, uint200, uint208, uint216, uint224, uint232, uint240, uint248, uint256, var, void, ether, finney, szabo, wei, days, hours, minutes, seconds, weeks, years},	
	keywordstyle=[2]\color{teal}\bfseries,
	keywords=[3]{block, blockhash, coinbase, difficulty, gaslimit, number, timestamp, msg, gas, sender, sig, value, now, tx, gasprice, origin, add, epochinc, get, setminus, emptyset},	
	keywordstyle=[3]\color{violet}\bfseries,
	identifierstyle=\color{black},
	sensitive=false,
	comment=[l]{//},
	morecomment=[s]{/*}{*/},
	commentstyle=\color{gray}\ttfamily,
	stringstyle=\color{red}\ttfamily,
	morestring=[b]',
	morestring=[b]"
}

\AtBeginDocument{%
  }

\begin{document}


\title{Improving Blockchain Scalability with the \setchain\ Data-type}


\author{Margarita Capretto}
\email{margarita.capretto@imdea.org}
\orcid{0000-0003-2329-3769}
\affiliation{%
  \institution{IMDEA Software Inst.}
  \city{Pozuelo de Alarcón}
  \state{Madrid}
  \country{Spain}
}
\affiliation{%
  \institution{Universidad Politécnica de Madrid}
  \state{Madrid}
  \country{Spain}
}

\author{Martín Ceresa}
\email{martin.ceresa@imdea.org}
\orcid{0000-0003-4691-5831}
\affiliation{%
  \institution{IMDEA Software Inst.}
  \city{Pozuelo de Alarcón}
  \state{Madrid}
  \country{Spain}
}

\author{Antonio {Fernández Anta}}
\email{antonio.fernandez@imdea.org}
\orcid{0000-0001-6501-2377}
\affiliation{%
  \institution{IMDEA Networks Inst.}
  \city{Leganés}
  \state{Madrid}
  \country{Spain}
}

\author{Antonio Russo}
\email{antonio.russo@imdea.org}
\orcid{0000-0003-3795-3000}
\affiliation{%
  \institution{IMDEA Networks Inst.}
  \city{Leganés}
  \state{Madrid}
  \country{Spain}
}
\affiliation{%
  \institution{Universidad Carlos III de Madrid}
  \city{Leganés}
  \state{Madrid}
  \country{Spain}
}

\author{César Sánchez}
\email{cesar.sanchez@imdea.org}
\orcid{0000-0003-3927-4773}
\affiliation{%
  \institution{IMDEA Software Inst.}
  \city{Pozuelo de Alarcón}
  \state{Madrid}
  \country{Spain}
}

\renewcommand{\shortauthors}{Capretto et al.}

\begin{abstract}
  %
  %
  Blockchain technologies are facing a scalability challenge, which
  must be overcome to guarantee a wider adoption of the technology.
  This scalability issue is due to the use of consensus algorithms to
  guarantee the total order of the chain of blocks (and of the
  transactions within each block).
  However, total order is often not fully necessary, since important
  advanced applications of smart-contracts do not require a total
  order among \emph{all} operations.
  A much higher scalability can potentially be achieved if a more
  relaxed order (instead of a total order) can be exploited.
  
  %
  %
  In this paper, we propose a novel distributed concurrent data type,
  called \emph{\setchain}, which improves scalability significantly.
  A \setchain implements a \emph{grow-only set} whose elements
  are not ordered, unlike conventional blockchain operations.
  When convenient, the \setchain allows forcing a synchronization
  barrier that assigns permanently an epoch number to a subset of the
  latest elements added, agreed by consensus.
  Therefore, two operations in the same epoch are not ordered, while
  two operations in different epochs are ordered by their respective
  epoch number.
  %
  %
  We present different Byzantine-tolerant implementations of
  \setchain, prove their correctness and report on an empirical
  evaluation of a prototype implementation.
  %
  %
  %
  Our results show that \setchain is orders of magnitude faster than
  consensus-based ledgers, since it implements grow-only sets with
  epoch synchronization instead of total order.
  
  Since the \setchain barriers can be synchronized with the underlying
  blockchain, \setchain objects can be used as a \emph{sidechain} to
  implement many decentralized solutions with much faster operations
  than direct implementations on top of blockchains.

  Finally, we also present an algorithm that encompasses in a single
  process the combined behavior of the Byzantine servers, which
  simplifies correctness proofs by encoding the general attacker in a
  concrete implementation.
\end{abstract}




\maketitle

\section{Introduction}
\label{sec:introduction}

\subsection{The Problem}

%
%
\emph{Distributed ledgers} (also known as \emph{blockchains}) were first proposed
by Nakamoto in 2009~\cite{nakamoto06bitcoin} in the implementation of Bitcoin, 
as a method to eliminate trustable third parties in electronic payment systems.
Modern blockchains incorporate smart
contracts~\cite{szabo96smart,wood2014ethereum}, which are immutable
state-full programs stored in the blockchain that describe
functionality of transactions, including the exchange of
cryptocurrency.
Smart contracts allow to describe sophisticated functionality, enabling
many applications in decentralized finances (DeFi)\footnote{As of December 2021, the monetary value locked in DeFi
  was estimated to be around \$100B, according to Statista
  \url{https://www.statista.com/statistics/1237821/defi-market-size-value-crypto-locked-usd/}.}, 
  decentralized governance,
Web3, etc.

The main element of all distributed ledgers is the ``blockchain,''
which is a distributed object that contains, packed in blocks, the
totally ordered list of transactions performed on behalf of the
users~\cite{anta2018formalizing,anta2021principles}.
The Blockchain object is maintained by multiple servers without a
central authority using consensus algorithms that are resilient to
Byzantine attacks.

%
%
A current major obstacle for a faster widespread adoption of
blockchain technologies is their limited scalability, due to the
limited throughput inherent to Byzantine consensus
algorithms~\cite{Croman2016ScalingDecentralizedBlockchain,Tyagi@BlockchainScalabilitySol}.
Ethereum~\cite{wood2014ethereum}, one of the most popular blockchains,
is limited to less than 4 blocks per minute, each containing less than
two thousand transactions.
Bitcoin~\cite{nakamoto06bitcoin} offers even lower throughput.
These figures are orders of magnitude slower than what many
decentralized applications require, and can ultimately jeopardize the
adoption of the technology in many promising domains.
This limit in the throughput increases the price per operation, due to
the high demand to execute operations.
Consequently, there is a growing interest in techniques to improve the
scalability of
blockchains~\cite{Zamani2018RapidChain,Xu2021SlimChain}.
%
Approaches include:
\begin{itemize}
\item developing faster consensus algorithms~\cite{Wang2019FastChain};
\item implementing parallel techniques, like
  sharding~\cite{Dang2019Sharding};
\item application-specific blockchains with Inter-Blockchain
  Communication capabilities \cite{wood2016polkadot,kwon2019cosmos};
\item executing smart contracts off-chain with the minimal required
  synchronization to preserve the guarantees of the blockchain---
  known as a ``layer 2'' (L2) approaches~\cite{Jourenko2019SoKAT}.
  Different L2 approaches are (1) the off-chain computation of
  Zero-Knowledge proofs~\cite{Sasson2014ZKvonNeumann}, which only need
  to be checked on-chain (hopefully more
  efficiently)~\cite{Sasson2014ZeroCash}, (2) the adoption of limited
  (but useful) functionality like \emph{channels} (e.g.,
  Lightning~\cite{Poon2016lightning}), or (3) the deployment of
  optimistic rollups (e.g., Arbitrum~\cite{Kalodner2018Arbitrum})
  based on avoiding running the contracts in the servers (except when
  needed to annotate claims and resolve disputes).
\end{itemize}

In this paper, we propose an alternative approach to increase blockchain
scalability that exploits the following observation.
%
%
It has been traditionally assumed that cryptocurrencies require 
total order to guarantee the absence of double-spending.
However, many useful applications and functionalities (including some
uses of cryptocurrencies~\cite{DBLP:conf/podc/GuerraouiKMPS19}) can
tolerate more relaxed guarantees, where operations are only
\emph{partially ordered}.
We propose here a Byzantine-fault tolerant implementation of a
distributed grow-only set \cite{Shapiro2011CCRDT,Cholvi2021BDSO},
equipped with an additional operation for introducing points of
barrier synchronization (where all servers agree on the contents of
the set).
Between barriers, elements of the distributed set can be temporarily
known by some but not all servers.
%
We call this distributed data structure \setchain.
A blockchain \(\mathcal{B}\) implementing \setchain (as well as
blocks) can align the consolidation of the blocks of \(\mathcal{B}\)
with barrier synchronizations, obtaining a very efficient set object as side
data type, with the same Byzantine-tolerance guarantees that \(\mathcal{B}\)
itself offers.

Two extreme implementations of sets with epochs in the context of
blockchains are:

\paragraph{A Completely off-chain implementation}\label{pa:offchain:sol}
The major drawback of having a completely off-chain implementation is
that from the point of view of the underlying blockchain the resulting
implementation does not have the trustability and accountability
guarantees that blockchains offer.
One example of this approach are \emph{mempools}.
Mempools (short for memory pools) are a P2P data type used by most
blockchains to maintain the set of pending transactions.
Mempools fulfill two objectives: (1) to prevent distributed
attacks to the servers that mine blocks and (2) to serve as a pool of
transaction requests from where block producers select operations.
Nowadays, mempools are receiving a lot of attention, since they suffer
from lack of accountability and are a source of
attacks~\cite{Saad2018DDoSMempool,Saad2019DDoSMempool}, including
front-running~\cite{Daian2020FlashBoys,Robinson2020DarkForest,Ferreira2021Frontrunner}.
Our proposed data structure, \epochs, offers a much stronger accountability,
because it is resilient to Byzantine attacks and the elements of the set that
\epochs maintains are public and cannot be forged.

\paragraph{Completely on-chain solution} 
Consider the following implementation (in a language similar to
Solidity), where \lstinline|add| is used to add elements,
and
\lstinline|epochinc| to increase epochs.
\begin{lstlisting}[language=Solidity,numbers=none]
  contract Epoch {
    uint public epoch = 0;
    set public the_set = emptyset;
    mapping(uint => set) public history;
    function add(elem data) public {
      the_set.add(data);
    }
    function epochinc() public {
      history[++epoch] = the_set.setminus(history);
    }
  }
\end{lstlisting}
%
%
One problem of this implementation is that every time we add an
element, \lstinline|the_set| gets bigger, which can affect the
required cost to execute the contract.
%
A second more important problem is that adding elements is
\emph{slow}---as slow as interacting with the blockchain---while our
main goal is to provide a much faster data structure than the
blockchain.

Our approach is faster, and can be deployed independently of the
underlying blockchain and synchronized with the blockchain
nodes.
Thus, \setchain\ lies between the two extremes described above.

For a given blockchain \(\mathcal{B}\), we propose an implementation
of \setchain that (1) is much more efficient than implementing and
executing operations directly in \(\mathcal{B}\);
(2) offers the same decentralized guarantees against
Byzantine attacks than \(\mathcal{B}\), and
(3) can be synchronized with the evolution of \(\mathcal{B}\), so contracts could
potentially inspect the contents of the \setchain.
In a nutshell, these goals are achieved by using faster operations for
the coordination among the servers for non-synchronized element
insertions, and using only consensus style algorithms for epoch
changes.

\subsection{Applications of \setchain}
\label{sec:motivation}
The potential applications that motivate the development of \epochs
include:

\subsubsection{Mempool}
Most blockchains store transaction requests from users in a
``mempool'' before they are chosen by miners, and once mined the
information from the mempool is lost.
Recording and studying the evolution of mempools would require an
additional object serving as a reliable mempool \emph{log system},
which must be fast enough to record every attempt of interaction with
the mempool without affecting the performance of the blockchain.
\setchain can server as such trustable log system, in this case
requiring no synchronization between epochs and blocks.

\subsubsection{Scalability by L2 Optimistic Rollups}
Optimistic rollups, like Arbitrum~\cite{Kalodner2018Arbitrum}, exploit
the fact that computation can be performed outside the blockchain,
posting on-chain only claims about the effects of the transactions.
In this manner Arbitrum maintainers propose the next state reached
after executing several transactions.
After some time, an arbitrator smart contract that is installed
on-chain assumes that a proposed step is correct because the state has
not been challenged, and executes the annotated effects.
A conflict resolution algorithm, also part of the contract on-chain, is
used to resolve disputes.
This protocol does not require a strict total order, but only a record
of the actions proposed.
Moreover, conflict resolution can be reduced to claim validation,
which could be performed by the maintainers of the \setchain,
removing the need for arbitration.

\subsubsection{Sidechain Data}
Finally, \setchain can also be used as a generic side-chain service
used to store and modify data in a manner that is synchronized with
the blocks.
Applications that require only to update information in the storage space of a
smart contract, like digital registries, can benefit from having faster (and
therefore cheaper) methods to manipulate the storage without invoking expensive
blockchain operations.

\subsection{Contributions.}

In summary, the contributions of the paper are the following:
\begin{compactitem}
\item the design and implementation of a side-chain data structure
  called \setchain,
\item several implementations of \setchain, providing different levels
  of abstraction and algorithmic implementation improvements,
\item an empirical evaluation of a prototype implementation, which
  suggests that \setchain is several orders of magnitude faster than
  consensus,
\item a description of how \setchain can be used as a distributed
  object, which requires for good clients to contact several servers
  for injecting elements and for obtaining a correct view of the
  \setchain;
\item a description of a much more efficient \setchain optimistic
  service that requires clients to contact only one server both
  for addition and for obtaining a correct state.
\item a reduction from the combined behavior of several Byzantine
  servers to a single non-deterministic process that simplifies
  reasoning about the combined distributed system.
\end{compactitem}

The rest of the paper is organized as follows.
Section~\ref{sec:prelim} contains preliminary model and assumptions.
Section~\ref{sec:apisol} describes the intended properties of~\setchain.
Section~\ref{sec:implementation} describes three different implementations of
\setchain where we follow an incremental approach.
Alg.~\ref{DPO-alg1} and~\ref{DPO-alg-basb} are required to explain how we
have arrived to Alg.~\ref{DPO-alg-basb3}, which is the fastest and most robust.
Section~\ref{sec:properties} proves the correctness of our three algortihms.
Section~\ref{sec:empirical} discusses an empirical evaluation of our
prototype implementations of the different algorithms.
Section~\ref{sec:wrappers} shows how to make the use of \epochs more
robust against Byzantine servers.
Section~\ref{sec:byzantine} presents a non-deterministic algorithm that
captures Byzantine behaviour. 
Finally, Section~\ref{sec:discussion} concludes the paper.


\section{Preliminaries}\label{sec:prelim}

We present now the model of computation and the building blocks
used in our \setchain algorithms.

\subsection{Model of Computation}\label{sec:model:computation}

A distributed system consists of processes---clients and
servers---with an underlying communication network with which each
process can communicate with every other process.
The communication is performed using message passing.
Each process computes independently and at its own speed, and the
internals of each process remain unknown to other processes.
Message transfer delays are arbitrary but finite and also remain
unknown to processes.
The intention is that servers communicate among themselves to
implement a distributed data type with certain guarantees, and clients
can communicate with servers to exercise the data type.

Processes can fail arbitrarily, but the number of failing (Byzantine)
servers is bounded by \(f\), and the total number of servers, \(n\),
is at least \(3f+1\).
We assume \emph{reliable channels} between non-Byzantine (correct)
processes, so no message is lost, duplicated or modified.
Each process (client or server) has a pair of public and private
keys.
The public keys have been distributed reliably to all the processes
that may interact with each other.
Therefore, we discard the possibility of spurious or fake processes.
We assume that messages are authenticated, so that messages corrupted
or fabricated by Byzantine processes are detected and discarded by
correct processes~\cite{Cristin1996AtomicBroadcast}.
As result, communication between correct processes is reliable but
asynchronous by default.
However, for the set consensus service that we use as a basic building
block, partial synchrony is
required~\cite{Crain2021RedBelly,Fischer1985Impossibility}, as
presented below.
Partial synchrony is only for the messages and computation of the
protocol implementing set consensus.
Finally, we assume that there is a mechanism for clients to create
``valid objects'' that servers can check locally.
In the context of blockchains this is implemented using public-key
cryptography.

\subsection{Building Blocks}
We use four building blocks to implement \setchains:

\subsubsection{Byzantine Reliable Broadcast (BRB)}
BRB
services~\cite{DBLP:journals/iandc/Bracha87,Raynal2018FaultTolerantMessagePassing}
allow to broadcast messages to a set of processes guaranteeing that
messages sent by correct processes are eventually received by
\emph{all} correct processes and all correct processes eventually
receive \emph{the same} set of messages.
A BRB service provides a primitive \(\BRBbroadcast(m)\) for sending messages and
an event \(\BRBdeliver(m)\) for receiving messages.
We list the relevant properties of BRB required to prove properties of
\setchain~(Section~\ref{sec:properties}):

\begin{compactitem}
\item\textbf{BRB-Validity:}\label{BRB-Validity} If a correct process
  \(p_{i}\) executes $\BRBdeliver(m)$ then $m$ was sent by a correct
  process \(p_{j}\) which executed $\BRBbroadcast(m)$ in the past.
\item\textbf{BRB-Termination(Local):}\label{BRB-Termination1} If a correct 
  process executes $\BRBbroadcast(m)$, then it executes $\BRBdeliver(m)$.
\item\textbf{BRB-Termination(Global):}\label{BRB-Termination2} If a
  correct process executes $\BRBdeliver(m)$, then all correct
  processes eventually execute $\BRBdeliver(m)$.
\end{compactitem}
Note that BRB services do not guarantee the delivery of messages in
the same order to two different correct participants.

\subsubsection{Byzantine Atomic Broadcast (BAB)}

BAB services~\cite{Defago2004BAB} extend BRB with an additional
guarantee: a total order of delivery of the messages.
BAB services provide the same operation and event as BRB, which we rename as
\(\BABbroadcast(m)\) and \(\BABdeliver(m)\).
However, in addition to validity and termination, BAB services also provide:
\begin{compactitem}
\item\textbf{Total Order:}\label{BAB-UniformTO} If two correct
  processes \(p\) and \(q\) both execute \(\BABdeliver(m)\) and \\
  \(\BABdeliver(m')\), then \(p\) delivers \(m\) before \(m'\) if and
  only if \(q\) delivers \(m\) before \(m'\).
\end{compactitem}
BAB has been proven to be as hard as consensus~\cite{Defago2004BAB},
and thus, is subject to the same
limitations~\cite{Fischer1985Impossibility}.

\subsubsection{Byzantine Distributed Grow-only Sets (DSO)~\cite{Cholvi2021BDSO}}
Sets are one of the most basic and fundamental data structures in
computer science, which typically include operations for adding and
removing elements.
Adding and removing operations do not commute, and thus, distributed
implementations require additional mechanisms to keep replicas
synchronized to prevent conflicting local states.
One solution is to allow only additions.
Hence, a grow-only set is a set in which elements can only be added
but not removed, which is implementable as a conflict-free replicated
data structure~\cite{Shapiro2011CCRDT}.

Let $A$ be an alphabet of values.
A grow-only set $\GS$ is a concurrent object maintaining an internal
set $\GS.S \subseteq A$ offering two operations for any process $p$:
\begin{compactitem}
  \item $\GS.\<add>(r):$ adds an element $r\in A$ to the set $\GS.S$.
  \item $\GS.\<get>():$ retrieves the internal set of elements \(\GS.S\).
\end{compactitem}
Initially, the set $\GS.S$ is empty.
A Byzantine distributed grow-only set object (DSO) is a concurrent
grow-only set implemented in a distributed manner tolerant to
Byzantine attacks~\cite{Cholvi2021BDSO}.
We list the properties relevant to \setchain~(Section~\ref{sec:properties}):
\begin{compactitem}
\item \textbf{Byzantine Completeness}: All \(\<get>()\) and \(\<add>(r)\)
  operations invoked by correct processes eventually complete.
\item \textbf{DSO-AddGet}: All \(\<add>(r)\) operations will
  eventually result in $r$ being in the set returned by all $\<get>()$. 
\item \textbf{DSO-GetAdd}: Each element $r$ returned by $\<get>()$
  was added using $\<add>(r)$ in the past.
\end{compactitem}

\subsubsection{Set Byzantine Consensus (SBC)}
SBC, introduced in
RedBelly~\cite{Crain2021RedBelly}, is a Byzantine-tolerant distributed
problem, similar to consensus.
In SBC, each participant proposes a set of elements (in the particular case of
RedBelly, a set of transactions).
After SBC finishes, all correct servers agree on a set of
valid elements which is guaranteed to be a subset of the union of
the proposed sets.
Intuitively, SBC efficiently 
runs binary consensus to agree on the sets proposed by each
participant, such that if the outcome is positive then the set
proposed is included in the final set consensus.
We list the properties relevant to \epochs~(Section~\ref{sec:properties}):
\begin{itemize}
\item \textbf{SBC-Termination}: every correct process eventually decides a
  set of elements.
\item \textbf{SBC-Agreement}: no two correct processes decide different
  sets of elements.
\item \textbf{SBC-Validity}: the decided set of transactions is a
  subset of the union of the proposed sets.
\item \textbf{SBC-Nontriviality}: if all processes are correct and
  propose an identical set, then this is the decided set.
\end{itemize}

The RedBelly algorithm~\cite{Crain2021RedBelly} solves SBC in a system
with partial synchrony: there is an unknown global stabilization time
after which communication is synchronous.
(Other SBC algorithms may have different partial synchrony assumptions.)
Then,~\cite{Crain2021RedBelly} proposes to use SBC to replace
consensus algorithms in blockchains, seeking to improve scalability,
because all transactions to be included in the next block can be
decided with one execution of the SBC algorithm.
In RedBelly every server computes the same block by applying a
deterministic function that totally orders the decided set of
transactions, removing invalid or conflicting transactions.

Our use of SBC is different from implementing a blockchain.
We use it to synchronize the barriers between local views of distributed
grow-only sets.
To guarantee that all elements are eventually assigned 
to epochs, we need the following property in the SBC service used.
\begin{itemize}
\item \textbf{SBC-Censorship-Resistance}: there is a time $\tau$ after which,
if the proposed sets of all correct processes 
contain the same element $e$, then $e$ will be in the decided set.
\end{itemize}
In RedBelly, this property holds because after the global
stabilization time, all set consensus rounds decide sets from correct
processes~\cite[Theorem 3]{Crain2021RedBelly}.


\section{The \setchain\ Distributed Data Structure}\label{sec:apisol}

%
A key concept of \epochs\ is the \emph{epoch} number, which is a
global counter that the distributed data structure maintains.
The synchronization barrier is realized as an epoch change: the epoch
number is increased and the elements in the grow-only set that have
not been assigned to a previous epoch are stamped with the new epoch
number.

\subsection{The Way of \setchain}

Before presenting the API of the \setchain, we reason about the expected path an
element should travel in the \setchain\ and the properties we want the structure
to have.
The main goal of the \setchain\ side-chain data structure is to
exploit the efficiency opportunity of the lack of order within a set
using a fast grow-only set so users can add records (uninterpreted
data), and thus, have evidence that such records exist.
Either periodically or intentionally, users trigger an epoch change
that creates an evidence of membership or a clear separation in the
evolution of the \setchain.
Therefore, the \setchain\ offers three methods: \(\<add>,\<get>\) and
\(\<epochinc>\).

In a centralized implementation, we expect to find inserted elements
immediately after issuing an \(\<add>\).
In other words, after an \(\<add>(e)\) we expect that the element \(e\)
is in the result of \(\<get>\).
In a distributed setting, such a restriction is too strong, and we
instead expect elements to eventually be in the set.
Additionally, implementations have to guarantee consistency between
different correct nodes, i.e. they cannot contradict each other.

\subsection{API and Server State of the \setchain}

We consider a universe $U$ of elements that client processes can
inject into the set.
We also assume that servers can locally validate an element $e \in U$.
A \textbf{\setchain} is a distributed data structure where a
collection of server nodes, \(\mathbb{D}\), maintain:
\begin{itemize}
\item a set $\<theset> \subseteq U$ of elements added;
\item a natural number $\<epoch> \in \mathbb{N}$;
\item a map $\<history> : [1..\<epoch>] \rightarrow \mathcal{P}(U)$
  describing sets of elements that have been stamped with an epoch
  number ($\mathcal{P}(U)$ denotes the power set of $U$).
\end{itemize}
Each server node $v \in \mathbb{D}$ supports three operations,
available to any client process:
\begin{itemize}
\item $v.\<add>(e)$: requests to add $e$ to $\<theset>$.
\item $v.\<get>()$: returns the values of $\<theset>$, $\<history>$,
  and $\<epoch>$, as perceived by $v$\footnote{In practice, we would have other
query operations since values returned by \(\<get>()\) operation may grow
large.}.
\item $v.\<epochinc>(h)$ triggers an epoch change (i.e., a
  synchronization barrier).  It must hold that $h=\<epoch>+1$.
\end{itemize}
Informally, a client process $p$ invokes a $v.\<get>()$ operation on
node $v$ to obtain $(S,H,h)$, which is $v$'s view of set $v.\<theset>$
and map $v.\<history>$, with domain $[1\ldots h]$.
Process $p$ invokes $v.\<add>(e)$ to insert a new element $e$ in
$v.\<theset>$, and $v.\<epochinc>(h+1)$ to request an epoch increment.
At server $v$, the set $v.\<theset>$ contains the knowledge of $v$
about elements that have been added, including those that have not
been assigned an epoch yet, while \(v.\<history>\) contains only those
elements that have been assigned an epoch.
A typical scenario is that an element \(e \in U\) is first perceived
by $v$ to be in $\<theset>$, to eventually be stamped and copied to
$\<history>$ in an epoch increment.
However, as we will see, some implementations allow other ways
to insert elements, in which $v$ gets to know $e$ for the first time
during an epoch change.
The operation $\<epochinc>()$ initiates the process of collecting
elements in $\<theset>$ at each node and collaboratively decide which
ones are stamped with the current epoch.

Initially, both $\<theset>$ and $\<history>$ are empty and $\<epoch> = 0$ in
every correct server.
Note that client processes can insert elements to $\<theset>$ through
$\<add>()$, but only servers decide how to update $\<history>$, which
client processes can only influence by invoking $\<epochinc>()$.

At a given point in time, the view of $\<theset>$ may differ from server to server.
%
The algorithms we propose only provide eventual consistency guarantees, as
defined in the next section.
\subsection{Desired Properties}
We specify now properties of correct implementations of \setchain.
We provide first a low-level specification that assumes that clients
interact with a \emph{correct} server.
Even though clients cannot be sure of whether the server they contact
is correct we will use these properties in Section~\ref{sec:wrappers}
to build two correct clients: (1) a pessimistic client that contacts
many servers to guarantee that sufficiently many are correct, and (2)
an optimistic client that contacts only one server (hoping it will be
a correct one) and can later check whether the operation was
successful.

We start by requiring from a \setchain that every $\<add>$, $\<get>$,
and $\<epochinc>$ operation issued on a correct server eventually
terminates.
We say that element $e$ is in epoch $i$ in history $H$ (e.g., returned by a $\<get>$
invocation) if $e \in H(i)$.
%
We say that element $e$ is in $H$ if there is an epoch $i$ such that
$e \in H(i)$.
%
%
The first property states that epochs only contain elements
coming from the grow-only set.
\newcounter{prop:consistent-set}
\setcounter{prop:consistent-set}{\value{property}}

\begin{property}[Consistent Sets]\label{api:consistent-set}
  Let $(S,H,h)=v.\<get>()$ be the result of an invocation to a correct server
$v$ where \(e\) is valid element.
  Then, for each $i\leq h, H(i) \subseteq S$.
\end{property}
The second property states that every element added to a correct server is
eventually returned in all future gets issued on the same server.
\begin{property}[Add-Get-Local]\label{api:history->theset-local}
  Let $v.\<add>(e)$ be an operation invoked on a correct server $v$.
  Then, eventually all invocations $(S,H,h)=v.\<get>()$ 
  satisfy $e\in S$.
\end{property}
The next property states that elements present
in a correct server are propagated to all correct servers.
\begin{property}[Get-Global]\label{api:history->theset}
  Let $v$ and $w$ be two correct servers, let $e \in U$ and let
  $(S,H,h)=v.\<get>()$.
  If $e \in S$, then eventually all invocations
  $(S',H',h')=w.\<get>()$ satisfy that $e \in S'$.
\end{property}
We assume in the rest of the paper that at every point in time there is a
future instant at which $\<epochinc>()$ is invoked and completed.
This is a reasonable assumption in any real practical scenario since
it can be easily guaranteed using timeouts.
Then, the following property states that all elements added are eventually
assigned an epoch.
\begin{property}[Eventual-Get]\label{api:theset->history}
  Let $v$ be a correct server, let $e \in U$ and let $(S,H,h)=v.\<get>()$.
  If $e \in S$, then eventually all invocations
  $(S',H',h')=v.\<get>()$ satisfy that $e \in H'$.
\end{property}
The previous three properties imply the following property.
\begin{property}[Get-After-Add]\label{api:get-after-add}
  Let $v.\<add>(e)$ be an operation invoked on a correct server $v$ with $e \in U$.
  Then, eventually all invocations $(S,H,h)=w.\<get>()$ on correct
  servers $w$ satisfy that $e\in H$.
\end{property}
%
%
An element can be in at most one epoch, and no element can be in two
different epochs even if the history sets are obtained from $\<get>$
invocations to two different (correct) servers.
\begin{property}[Unique Epoch]\label{api:local_unique_stamp}
  Let $v$ be a correct server,
  $(S,H,h)=v.\<get>()$, and let
  $i,i'\leq{}h$ with $i\neq i'$.
  Then, $H(i)\cap{}H(i')=\emptyset$.
\end{property}
All correct server processes agree on the epoch contents.
\begin{property}[Consistent Gets]\label{api:consistent-gets}
  Let $v,w$ be correct servers, let $(S,H,h)=v.\<get>()$ and
  $(S',H',h')=w.\<get>()$, and let $i\leq \min(h,h')$. Then
  $H(i)=H'(i)$.
\end{property}
Property~\ref{api:consistent-gets} states that the histories returned
by two $\<get>$ invocations to correct servers are one the prefix of
the other.
However, since two elements $e$ and $e'$ can be inserted at two
different correct servers---which can take time to propagate---, the
$\<theset>$ part of $\<get>$ obtained from two correct servers may not
be contained in one another.

Finally, we require that every element in the history comes from the
result of a client adding the element.
%
%
%
\begin{property}[Add-before-Get]\label{api:get->add}
  Let $v$ be a correct server, $(S,H,h)=v.\<get>()$,
  and $e \in S$.
  Then, there was an operation $w.\<add>(e)$ in the past in some
  server $w$.
\end{property}

Properties~\ref{api:consistent-set}, \ref{api:local_unique_stamp},
\ref{api:consistent-gets} and \ref{api:get->add} are safety
properties.
Properties~\ref{api:history->theset-local}, \ref{api:history->theset},
\ref{api:theset->history} and \ref{api:get-after-add} are liveness
properties.


\section{Implementations}\label{sec:implementation}

In this section, we describe implementations of \setchain that satisfy
the properties in Section~\ref{sec:apisol}.
We describe a centralized sequential implementation to build up intuition, and
then three distributed implementations.
The first distributed implementation is built using a Byzantine  distributed
grow-only set object (DSO) to maintain $\<theset>$, and Byzantine atomic
broadcast (BAB) for epoch increments.
The second distributed implementation is also built using DSO, but we replace
BAB with Byzantine reliable broadcast (BRB) to announce epoch increments and set
Byzantine consensus (SBC) for epoch changes.
Finally, we replace DSO with local sets, and we use BRB for
broadcasting elements and epoch increment announcements, and SBC for
epoch changes, resulting the fastest implementation.

\subsection{Sequential Implementation}

\addtocounter{algorithm}{-1} 
\begin{figure}[t!]
  \begin{adjustbox}{minipage=[t]{\columnwidth}}
    \begin{algorithm}[H]
      \renewcommand{\thealgorithm}{Central}         
         \caption{\small Single server implementation.}%
        \label{alg:sequential}%
              \small
              \begin{algorithmic}[1]
                \State \textbf{Init:} $\<epoch> \leftarrow 0,$
                \hspace{2em}
                $\<history> \leftarrow \emptyset$\label{seq:history} 
                \State \textbf{Init:} $\<theset> \leftarrow \emptyset$\label{seq:theset}
                \Function{\<Get>}{~}
                  \State \textbf{return} $(\<theset>, \<history>, \<epoch>)$
                \EndFunction
                \Function{\<Add>}{$e$}
                  \State \textbf{assert} {\(valid(e)\)}
                  \State \(\<theset> <- \<theset> \cup \{e\}\)
                \EndFunction
                \Function{\<EpochInc>}{$h$}
                    \State \textbf{assert} $h \equiv \<epoch> + 1$
                    \State \(\textit{proposal} <- \<theset> \setminus \bigcup_{k=1}^{\<epoch>} \<history>(k)\)
                    \State $\<history> \leftarrow \<history> \cup \{\langle h, \textit{proposal} \rangle\}$
                    \State $\<epoch> \leftarrow \<epoch> + 1$
                \EndFunction
              \end{algorithmic}
            \end{algorithm}
      \end{adjustbox}
  \end{figure}

Alg.~\ref{alg:sequential} shows a centralized solution, which maintains two
local sets, $\<theset>$---to record added elements---, and $\<history>$, which keeps
a collection of pairs $\langle h,A \rangle$ where $h$ is an epoch number and $A$ is a set of
elements.
We use $\<history>(h)$ to refer to the set $A$ in the pair
$\langle h,A \rangle \in \<history>$.
A natural number $\<epoch>$ is incremented each time there is a new
epoch.
The operations are:
$\<Add>(e)$, which checks that element $e$ is valid and adds it to
$\<theset>$, and
$\<Get>()$, which returns $(\<theset>,\<history>,\<epoch>)$.
%

There is only one way to add elements, through the use of operation
\(\<Add>()\).
Since Alg.~\ref{alg:sequential} does not maintain a distributed data structure,
but a centralized one, there are no Byzantine nodes.
Therefore, every time clients interact with the only server, which is
correct.
The following implementations are distributed so they must incorporate
some mechanism to prevent Byzantine nodes from polluting or
manipulating the \setchain.

\begin{figure}[t!]
  \begin{adjustbox}{minipage=[t]{\columnwidth}}
    \begin{algorithm}[H]
      \renewcommand{\thealgorithm}{Basic}         
      \caption{\small Server $i$ implementation using DSO and BAB}%
      \label{DPO-alg1}%
      \small
      \begin{algorithmic}[1]
                \State \textbf{Init:} $\<epoch> \leftarrow 0,\hspace{2em}\<history>\leftarrow\emptyset$\label{alg1:history}
            \State \textbf{Init:} $\<theset> \leftarrow $\<DSO>.\<Init>\(()\)~\label{alg1:init_e}
            \Function{\<Get>}{~}
              \State \<return> $(\<theset>.\text{\<Get>}()\cup\<history>, \<history>,\<epoch>)$~\label{alg1:func_get}
            \EndFunction
            \Function{\<Add>}{$e$}
            \State \textbf{assert} {\(\<valid>(e)\)}
            \State $\<theset>.\text{\<Add>}(e)$
            \EndFunction
            \Function{\<EpochInc>}{$h$}\label{alg1:epoch_inc}
              \State \textbf{assert} $h \equiv \<epoch> + 1$
              \State $\<proposal> \leftarrow
                \<theset>.\text{\<Get>}()
                \setminus \bigcup_{k=1}^{\<epoch>} \<history>(k)$~\label{alg1:proposal}
              \State \<BAB>.\<Broadcast>(\<mepochinc>($h, \<proposal>, i$))\label{alg1:bcast_proposal}
            \EndFunction
            \Upon{\<BAB>.\<Deliver>(\<mepochinc>($h, \<proposal>,
                  j$))\\\hspace{2em} from $2f+1$ different servers $j$ for the same $h$}\label{alg1:upon_delivery} 
              \State \textbf{assert} $h \equiv \<epoch> + 1$
              \State $E \leftarrow \{e:e \in \<proposal>$ for at least $f+1$ different j\}
              \State $\<history> \leftarrow \<history> \cup \{\langle h, E \rangle\}$\label{alg1:added}
              \State $\<epoch> \leftarrow \<epoch> + 1$
            \EndUpon
          \end{algorithmic}
        \end{algorithm}
        \end{adjustbox}
  \end{figure}

\subsection{Distributed Implementations}

\subsubsection{First approach. DSO and BAB}
Alg.~\ref{DPO-alg1} uses two external services: DSO and BAB.
%
We denote messages with the name of the message followed by its
content as in ``$\<mepochinc>(h,\textit{proposal},i)$''.
The variable $\<theset>$ is not a local set anymore, but a DSO initialized
empty with $\<Init>()$ in line~\ref{alg1:init_e}.
The function $\<Get>()$ invokes the DSO $\<Get>()$ function
(line~\ref{alg1:func_get}) to fetch the set of elements.
The function $\<EpochInc>(h)$ triggers the mechanism required to increment an
epoch and reach a consensus on the elements belonging to epoch \(h\).
%
The consensus process begins by computing a local $\textit{proposal}$ set, of
those elements added but not stamped (line~\ref{alg1:proposal}).
The $\textit{proposal}$ set is then broadcasted using a BAB service
alongside the epoch number $h$ and the server node id $i$
(line~\ref{alg1:bcast_proposal}).
Then, the server waits to receive exactly $2f+1$ proposals, and keeps
the set of elements $E$ present in at least $f+1$ proposals, which
guarantees that each element $e \in E$ was proposed by at least one correct
server.
The use of BAB guarantees that every message sent by a correct server
eventually reaches every other correct server in the \emph{same order}, so
all correct servers have the same set of $2f+1$ proposals.
Therefore, all correct servers arrive to the same conclusion, 
and the set $E$
is added as epoch $h$ in $\<history>$ in line~\ref{alg1:added}.

Alg.~\ref{DPO-alg1}, while easy to understand and prove correct, is not
efficient.
First, in order to complete an epoch increment, it requires at
least $3f+1$ calls to $\<EpochInc>(h)$ to different servers, so at
least $2f+1$ proposals are received (the $f$ Byzantine severs may not
propose anything).
Another source of inefficiency comes from the use of off-the-shelf
building blocks.
For instance, every time a DSO $\<Get>()$ is invoked, many messages
are exchanged to compute a reliable local view of the
set~\cite{Cholvi2021BDSO}.
Similarly, every epoch change requires a DSO $\<Get>()$ in
line~\ref{alg1:proposal} to create a proposal.
Additionally, line~\ref{alg1:upon_delivery} requires waiting for
$2f+1$ atomic broadcast deliveries to take place.
The most natural implementations of BAB services solve one
consensus per message delivered (see Fig.~7 in
\cite{Chandra1998FailureBAB}), which would make this algorithm very
slow.
We solve these problems in two alternative algorithms.
\begin{figure}[t!]
      \begin{adjustbox}{minipage=[t]{\columnwidth}}
  \begin{algorithm}[H]
    \renewcommand{\thealgorithm}{Slow}         
    \caption{\small Server $i$ implementation using DSO, BRB and SBC.
       }
    \label{DPO-alg-basb}
    \small
    \begin{algorithmic}[1]
      \setcounter{ALG@line}{6} 
      \State \ldots     \Comment{$\<Get>$ and $\<Add>$ as in Alg.~\ref{DPO-alg1}}

      \Function{\<EpochInc>}{$h$}\label{alg2:epochinc}
        \State \<assert> $h==\EPOCH +1$
        \State \<BRB>.\<Broadcast>(\<mepochinc>($h$))\label{alg2:brb-epochinc}
      \EndFunction
      \Upon{\<BRB>.\<Deliver>(\<mepochinc>($h$)) and $h<\EPOCH+1$}
        \State \<drop>
      \EndUpon
      \Upon{\<BRB>.\<Deliver>(\<mepochinc>($h$)) and $h==\EPOCH+1$}
        \State \<assert> $prop[h]==null$
        \State $prop[h] \leftarrow \THESET.\text{\<Get>}() \setminus 
        \bigcup_{k=1}^{\<epoch>} \<history>(k)$\label{DPO-alg2-dsoget}
        \State \<SBC>[$h$].\<Propose>($prop[h]$)
      \EndUpon
      %
      %
      \Upon{\<SBC>[$h$].\<SetDeliver>($propset$)  and  $h==\EPOCH+1$} 
         \State $E \leftarrow \{e : e \in propset[j], valid(e) \wedge e\notin \HISTORY\}$
         \State $\THESET \leftarrow \THESET.Add(E)$
         \State $\HISTORY \leftarrow \HISTORY \cup \{\langle h, E \rangle\}$
        \State{$\EPOCH \leftarrow \EPOCH+1$} 
        \tikz{\node(tmp2){};}
        \EndUpon
      \end{algorithmic}
    \end{algorithm}
        \end{adjustbox}
\end{figure}

\begin{figure}[t!]
	\begin{algorithm}[H]
		\renewcommand{\thealgorithm}{Fast}         	
          \caption{\small Server implementation using a local set,
            BRB and SBC.
            }~\label{DPO-alg-basb3}~\label{DPO-alg3}
          \small
		\begin{algorithmic}[1]
			\State  \textbf{Init:} $\EPOCH \leftarrow 0,$
			\hspace{2em} $\HISTORY \leftarrow \emptyset$ 
			\State  \textbf{Init:} $\THESET \leftarrow \emptyset$\label{alg3:theset}
			\Function{\<Get>}{~}
			        \State \<return> $(\THESET, \HISTORY,\<epoch>)$\label{alg3:get}
			\EndFunction
			\Function{\<Add>}{$e$} \label{alg3-add-begin}
                                \State \<assert> $valid(e)$ and $e \notin \THESET$
                                \State \<BRB>.\<Broadcast>(\<madd>(e)) \label{alg3:brb-broadcast}
                                \EndFunction \label{alg3-add-end}
			\Upon{\<BRB>.\<Deliver>(\<madd>\((e)\))} \label{alg3-BRBDeliver-add-begin}
                                \State \<assert> $valid(e)$
			        \State $\THESET \leftarrow \THESET \cup \{e\}$ \label{add_theset_BRBDeliver}
			\EndUpon \label{alg3-BRBDeliver-add-end}
			\Function{\<EpochInc>}{$h$} \label{alg3-epochinc-begin}
                                \State \<assert> $h==\EPOCH +1$
			        \State \<BRB>.\<Broadcast>(\<mepochinc>($h$))
                        \EndFunction \label{alg3-epochinc-end}
			\Upon{\<BRB>.\<Deliver>(\<mepochinc>($h$)) and
                          $h<\EPOCH+1$} \label{alg3-BRBDeliver-epinc-drop-begin}
                                \State \<drop>
                        \EndUpon \label{alg3-BRBDeliver-epinc-drop-end}
			\Upon{\<BRB>.\<Deliver>(\<mepochinc>($h$)) and
                          $h==\EPOCH+1$} \label{alg3-BRBDeliver-epinc-begin}
                                \State \<assert> \(prop[h] == \emptyset\)
                                \State $prop[h] \leftarrow \THESET \setminus 
                                \bigcup_{k=1}^{\<epoch>} \<history>(k)$
                                \State \<SBC>[$h$].\<Propose>($prop[h]$)
			\EndUpon \label{alg3-BRBDeliver-epinc-end}
			\Upon{\<SBC>[\(h\)].\<SetDeliver>($propset$)
                          and  $h==\EPOCH+1$} \label{alg3-SetDeliver-begin}
			        \State $E \leftarrow \{e : e \in propset[j], valid(e) \wedge e\notin \HISTORY\}$~\label{forall_e_SBCDeliver}

				    \State $\THESET \leftarrow \THESET \cup E$~\label{add_theset_SBCSetDeliver}
				    \State $\HISTORY \leftarrow \HISTORY \cup \{\langle h, E \rangle\}$~\label{add_history_SBCSetDeliver}
				    \State $\EPOCH \leftarrow \EPOCH+1$\label{set_epoch_SBCSetDeliver} 
			\EndUpon \label{alg3-SetDeliver-end}
		\end{algorithmic}
	\end{algorithm}\vspace{-2em}
\end{figure}

\subsubsection{Second approach. Avoiding BAB}

Alg.~\ref{DPO-alg-basb} improves the performance of
Alg.~\ref{DPO-alg1} as follows.
First, it uses BRB to propagate epoch increments.
Second, the use of BAB and the wait for the arrival of
$2f+1$ messages in line~\ref{alg1:upon_delivery} of
Alg.~\ref{DPO-alg1} is replaced by using a SBC
algorithm, which allows solving several consensus instances
simultaneously.

Ideally, when an $\<EpochInc>(h)$ is triggered unstampped elements in the
local $\<theset>$ of each correct server should be stamped with the
new epoch number and added to the set $\<history>$.
%
%
However, we need to guarantee that for every epoch the set
$\<history>$ is the same in every correct server.
Alg~\ref{DPO-alg1} enforced this using BAB and counting sufficient
received messages.
Alg.~\ref{DPO-alg-basb} uses SBC to solve several independent
consensus instances simultaneously, one on each participant's
proposal.
%
Line~\ref{alg2:brb-epochinc} broadcasts an invitation to an epoch
change, which causes correct servers to build a proposed set and
propose this set using the SBC.
There is one instance of SBC per epoch change $h$, identified by SBC$[h]$.
The SBC service guarantees that each correct server decides the same
set of proposals (where each proposal is a set of elements).
Then, every node applies the same function to the same set of
proposals reaching the same conclusion on how to update $\<history>(h)$.
The function preserves elements that are valid and unstampped.
Note that this opens the opportunity to add elements directly by
proposing them during an epoch change without broadcasting them
before.
This optimization is exploited in Section~\ref{sec:empirical} to speed
up the algorithm further.
As a final note, Alg.~\ref{DPO-alg-basb} allows a Byzantine server to
bypass $\<Add>()$ to propose elements, which will be accepted as long as
the elements are valid.
This is equivalent to a client proposing an element using an
$\<Add>()$ operation, which is then successfully propagated in an epoch
change.
Observe that Alg.~\ref{DPO-alg-basb} still triggers one invocation of
the DSO $\<Get>()$ at each server to build the local proposal.


\subsubsection{Final approach. BRB and SBC without DSOs}

Alg.~\ref{DPO-alg-basb3} avoids the cascade of messages that DSO
$\<Get>()$ calls require by dissecting the internals of the DSO, and
incorporating the internal steps in the \setchain algorithm directly.
This idea exploits the fact that \emph{a correct} \setchain
\emph{server} is a \emph{correct client} of the DSO, and there is no
need for the DSO to be defensive.
This observation shows that using Byzantine resilient building blocks
do not compose efficiently, but exploring this general idea is out of
the scope of this paper.

Alg.~\ref{DPO-alg-basb3} implements $\<theset>$ using a local set
(line~\ref{alg3:theset}).
Elements received in $\<Add>(e)$ are propagated using BRB.
At any given point in time two correct servers may have a different
local sets (due to pending BRB deliveries) but each element added in
one server will eventually be known to all others.
The local variable $\<history>$ is only updated in
line~\ref{add_history_SBCSetDeliver} as a result of a SBC round.
Therefore, all correct servers will agree on the same sets formed by
unstamped elements proposed by some servers.
%
%
Additionally, Alg.~\ref{DPO-alg-basb3} updates $\<theset>$ to
account for elements that are new to the server (line \ref{add_theset_SBCSetDeliver}) , guaranteeing that all
elements in $\<history>$ are also in $\<theset>$.
%


\section{Proof of Correctness}
\label{sec:properties}

We prove now the correctness of the distributed algorithms presented
in Section~\ref{sec:implementation} with regard to the desired
properties introduced in Section~\ref{sec:apisol}.

The following lemmas reason about how elements are stamped.
They directly imply Properties~\ref{api:theset->history}
(\textit{Eventual-Get}), \ref{api:local_unique_stamp} (\textit{Unique
  Epoch}) and \ref{api:consistent-gets} (\textit{Consistent Gets}),
respectively.
We prove that our algorithms satisfy these lemmas in the following
subsections.

\begin{lemma}\label{lem:eventually-history}
  Let $v$ and $w$ be correct servers. If $e\in v.\<theset>$. Then,
  eventually $e$ is in $w.\<history>$.
\end{lemma}

\begin{lemma}\label{lem:stamp-once}
  Let $v$ be a correct server and $h$, $h'$ two epoch numbers with
  $h\neq h'$. If $e\in v.\<history>(h)$ then
  $e\notin v.\<history>(h')$.
\end{lemma}

\begin{lemma}\label{lem:history-unique}
  Let $v$ and $w$ be correct servers.
  Let $h$ be such that $h\leq v.\<epoch>$ and $h\leq w.\<epoch>$. Then $v.\<history>(h)=w.\<history>(h)$.
\end{lemma}

\noindent It is easy to see that Property~\ref{api:get-after-add}
(\textit{Get-After-Add}) follows directly from
Properties~\ref{api:history->theset-local}, \ref{api:history->theset}
and \ref{api:theset->history}.

\subsection{Correctness of Alg.~\ref{DPO-alg1}}




Property~\ref{api:consistent-set} (\textit{Consistent Sets}) trivially
holds for Alg.~\ref{DPO-alg1} as the response for a $\<Get>()$ is the
tuple $(\<theset>.\<get>()\cup\<history>, \<history>,\<epoch>)$ (see
line~\ref{alg1:func_get}).
%
%
Properties~\ref{api:history->theset-local}
(\textit{Add-Get-Local}), \ref{api:history->theset}
(\textit{Get-Global}) and \ref{api:get->add} (\textit{Add-before-Get})
follow directly from properties \textbf{DSO-AddGet} and
\textbf{DSO-GetAdd}, as
all these properties reason about how elements are added to $\<theset>$
which is implemented as a DSO.

We prove now Lemmas~\ref{lem:eventually-history},~\ref{lem:stamp-once}
and~\ref{lem:history-unique}.
We start by proving Lemma~\ref{lem:history-unique}, which states that
two correct servers $v$ and $w$ agree on the content of epoch $h$
after both have reached $h$.

\begin{proposition}
  Lemmas~\ref{lem:eventually-history},~\ref{lem:stamp-once}
  and~\ref{lem:history-unique} hold for Alg.~\ref{DPO-alg1}.
\end{proposition}

\begin{proof}
  We prove first Lemma~\ref{lem:history-unique}, then
  Lemma~\ref{lem:eventually-history} and finally
  Lemma~\ref{lem:stamp-once}.
  \begin{itemize}
  \item We start by proving Lemma~\ref{lem:history-unique}, which
    states that two correct servers $v$ and $w$ agree on the content
    of epoch $h$ after both have reached $h$.
    Variable \<history> is updated only at line~\ref{alg1:added}.
    Since $v$ and $w$ are correct servers that computed
    $\<history>(h)$ both must have BRB.Deliver messages of the form
    \<mepochinc>($h, proposal,j$) from $2f+1$ different servers.
    The properties \textbf{BRB-Termination(Global)} and \textbf{Total
      Order} guarantee that these are the same messages for both
    servers.
    Both $v$ and $w$ filter elements in the same way, by keeping just
    the elements proposed by $f+1$ servers in one of the delivered
    messages.
    Therefore, in line~\ref{alg1:added} both $v$ and $w$ update
    $\<history>(h)$ equally.
    Hence, Lemma~\ref{lem:history-unique} holds.
  \item We prove now Lemma~\ref{lem:eventually-history}.
    Let \(e\) be an element in the \<theset> of a correct server.
    It follows from Property~\ref{api:history->theset} that there is a
    point in time $t$ after which $e$ is in the set returned by all
    $\<theset>.\<get>()$ in all correct servers.
    Note that we are assuming that there is always a new epoch
    increment, in particular there is a new $\<EpochInc>(h)$ after
    $t$.
    If $e$ is already in $w.\<history>(h')$ for $h'<h$ we are done.
    Otherwise, by Lemma~\ref{lem:history-unique}, $e$ will not be in
    $z.\<history>(h')$ for any correct servers $z$ and $h'<h$.
    Then, when computing the proposal for epoch $h$, at
    line~\ref{alg1:proposal}, all correct servers will include $e$ in
    their set.
    To compute the epoch $h$ server $w$ will wait until it BRB.Deliver
    messages of the form $\<mepochinc>(h, proposal,j)$ from $2f+1$
    different servers (see line~\ref{alg1:upon_delivery}), which will
    include at least $f+1$ messages from correct servers.
    Therefore, at least $f+1$ of the proposal received will contain
    $e$, which implies that server $w$ will include $e$ in the set
    $\<history>(h)$.
    Hence, Lemma~\ref{lem:eventually-history} holds.
  \item Finally, we use Lemma~\ref{lem:history-unique} to prove
    Lemma~\ref{lem:stamp-once}.
    The proof proceeds by contradiction.
  Assume that $e \in v.\<history>(h')$.
  The only point where $\<history>(h')$ is updated is at
  line~\ref{alg1:added}, where $\<history>(h')$ is defined to contain
  exactly the elements that were proposed for epoch $h'$ by $f+1$
  different servers.
  Therefore, given that there are at most $f$ Byzantine servers, at
  least one correct server proposed a set containing $e$ for epoch
  $h'$.
  Let $w$ be one of such correct server.
  Server $w$ computed the proposal for epoch $h'$
  (see line~\ref{alg1:proposal}) as the set $\<theset>.\text{\<Get>}()
  \setminus \bigcup_{k=1}^{h'-1} \<history>(k)$.
 But, by Lemma~\ref{lem:history-unique}, $w.\<history>(h) =
 v.\<history>(h)$, which implies that $e \in \bigcup_{k=1}^{h'-1}
 \<history>(k)$, as $h \leq h'-1$.
 Therefore, $e \notin \<theset>.\text{\<Get>}() \setminus
 \bigcup_{k=1}^{h'-1} \<history>(k)$, meaning that $e$ was not
 proposed by $w$ for epoch $h'$.
  This is a contradiction that follows from assuming that $e \in
  v.\<history>(h')$.
  Hence, Lemma~\ref{lem:stamp-once} holds for Alg.~\ref{DPO-alg1}.
\end{itemize}  
\end{proof}

\subsection{Correctness of Alg.~\ref{DPO-alg-basb}}

We can prove that Alg.~\ref{DPO-alg-basb} satisfies all lemmas using
the exactly same reasoning that we used for Alg.~\ref{DPO-alg1},
replacing properties of BRB and BAB by corresponding properties of
SBC.
In order to prove Lemma~\ref{lem:eventually-history}, instead of rely
of enough messages being BRB.Delivered, we use
\textbf{SBC-Censorship-Resistance} which guarantees that the decided
set will contain the element $e$.
The proof of Lemma~\ref{lem:history-unique} is analogous to the one of
Alg.\ref{DPO-alg1} but instead of using the properties
\textbf{BRB-Termination(Global)} and \textbf{Total Order}, we
use the property \textbf{SBC-Agreement} which guarantees that correct
servers agree on the decided set.

\subsection{Correctness of Alg.~\ref{DPO-alg-basb3}}
We now prove correctness of Alg.~\ref{DPO-alg-basb3}.
First, Property~\ref{api:history->theset-local}
(\textit{Add-Get-Local}) follows directly from the code of $\<Add>()$,
line~\ref{alg3:get} of $\<Get>()$ and Property
\textbf{BRB-Termination(Local)}.
Second, we show that all stamped elements are in $\<theset>$, which
implies Property~\ref{api:consistent-set} (\textit{Consistent Sets}).

\begin{lemma}
  \label{lem:consistent-set}
  At the end of each function and upon of Alg.~\ref{DPO-alg-basb3},
  $\bigcup_{h} v.\<history>(h) \subseteq v.\<theset>$, for every
  correct server $v$.
\end{lemma}

\begin{proof}
  Let $v$ be a correct server.
  The only way to add elements to $v.\<history>$ is at
  line~\ref{add_history_SBCSetDeliver}, which is preceded by
  line~\ref{add_theset_SBCSetDeliver} which adds the same elements to
  $v.\<theset>$.
  The only other instruction that modifies $v.\<theset>$ is
  line~\ref{add_theset_BRBDeliver} which only makes the set grow.
\end{proof}

Third, the following lemma shows that elements in a correct server are
propagated to all correct servers, which is equivalent to
Property~\ref{api:history->theset} (\textit{Get-Global}).

\begin{lemma}\label{lem:eventually-theset}
  Let $v$ be a correct server and $e$ an element in $v.\<theset>$.
  Then $e$ will eventually be in $w.\<theset>$ for every correct
  server $w$.
\end{lemma}
\begin{proof}
  Initially, $v.\<theset>$ is empty. There are two ways to add an
  element $e$ to $v.\<theset>$:
  (1) At line~\ref{add_theset_BRBDeliver}, so $e$ is valid and was
  received via a $\<BRB>.\<Deliver>(\<madd>(e))$.
    By Property \textbf{BRB-Termination\\(Global)}, every correct server $w$ will
    eventually execute $\<BRB>.\<Deliver>(\<madd>(e))$, and then
    (since $e$ is valid), $w$ will add it to $w.\<theset>$ in
    line~\ref{add_theset_BRBDeliver}.
  (2) At line~\ref{add_theset_SBCSetDeliver}, so element $e$ is valid
    and was received as an element in one of the sets in $\Propset$
    from $\<SBC>[h].\<SetDeliver>(\Propset)$ with $h = v.\<epoch> +
    1$.
    By properties \textbf{SBC-Termination} and \textbf{SBC-Agreement},
    all correct servers agree
    on the same set of proposals.
    Then, $w$ will also eventually receive
    $\<SBC>[h].\<SetDeliver>(\Propset)$.
    Therefore, if $v$ adds $e$ then $w$ either adds it or has it
    already in its $w.\<history>$ which implies by
    Lemma~\ref{lem:consistent-set} that $e\in w.\<theset>$.
    In either case, $e$ will eventually be in $w.\<theset>$.
\end{proof}

Finally, we prove the lemmas introduced at the beginning of the
section that reason about how elements are stamped.

\begin{proposition}
  Lemma~\ref{lem:eventually-history}, Lemma~\ref{lem:stamp-once}
  and Lemma~\ref{lem:history-unique} hold for Alg.~\ref{DPO-alg-basb3}.
\end{proposition}

\begin{proof}
  We prove each Lemma separately.
  \begin{itemize}
  \item The proof of Lemma~\ref{lem:eventually-history} (elements
    added in a correct server will eventually be stamped in all
    correct servers) for this algorithm is analogous to the one for
    Alg.~\ref{DPO-alg-basb}.
  \item Lemma~\ref{lem:stamp-once} follows directly from the check that $e$
    is not injected at $v.\<history>(h)$ if $e\in v.\<history>$ in
    line~\ref{add_history_SBCSetDeliver}.
  \item Next, we show that two correct servers $v$ and $w$ agree on
    the prefix of \<history> that both computed, that is
    Lemma~\ref{lem:history-unique}.
    The proof proceeds by induction on the epoch number $\<epoch>$.
    The base case is $\<epoch>=0$, which holds trivially since
    $v.\<history>(0)=w.\<history>(0)=\emptyset$.
    Variable $\<epoch>$ is only incremented in one unit in
    line~\ref{set_epoch_SBCSetDeliver}, after $\<history>(h)$ has been
    changed in line~\ref{add_history_SBCSetDeliver} when
    $h=\<epoch>+1$.  In that line, $v$ and $w$ are in the same phase
    on SBC (for the same $h$).
    By \textbf{SBC-Agreement}, $v$ and $w$ receive the same \Propset,
    both $v$ and $w$ validate all elements equally, and (by inductive
    hypothesis), for each $h' \leq \<epoch>$ it holds that
    $e\in v.\<history>(h')$ if and only if $e\in w.\<history>(h')$.
    Therefore, in line~\ref{add_history_SBCSetDeliver} both $v$ and
    $w$ update $\<history>(h)$ equally, and after
    line~\ref{set_epoch_SBCSetDeliver} it holds that
    $v.\<history>(\<epoch>)=w.\<history>(\<epoch>)$.
    \end{itemize}
  \end{proof}

Finally, we discuss Property~\ref{api:get->add}
(\textit{Add-before-Get}).
If valid elements can only be created by clients and added using
$\<Add>(e)$ the property trivially holds.
If valid elements can be created by, for example Byzantine servers,
then they can inject elements in $\<theset>$ and $\<history>$ of
correct servers without using $\<Add>()$.
These servers can either execute directly a BRB.$\<Broadcast>$ or
directly via the SBC in epoch rounds.
In these cases, Alg.~\ref{DPO-alg-basb3} satisfies a weaker version of
(\textit{Add-before-Get}) that states that elements returned by
$\<Get>()$ are either added by $\<Add>()$, by a BRB.$\<Broadcast>$ or
injected in the SBC phase.


\section{Empirical Evaluation}
\label{sec:empirical}

\newcommand{\Exp}[1]{\includegraphics[width=0.410\textwidth]{experiments/#1.pdf}}

\begin{figure*}[t!]
  \footnotesize\centering
  \noindent
   \begin{tabular}{@{}c@{}c@{}}
     \Exp{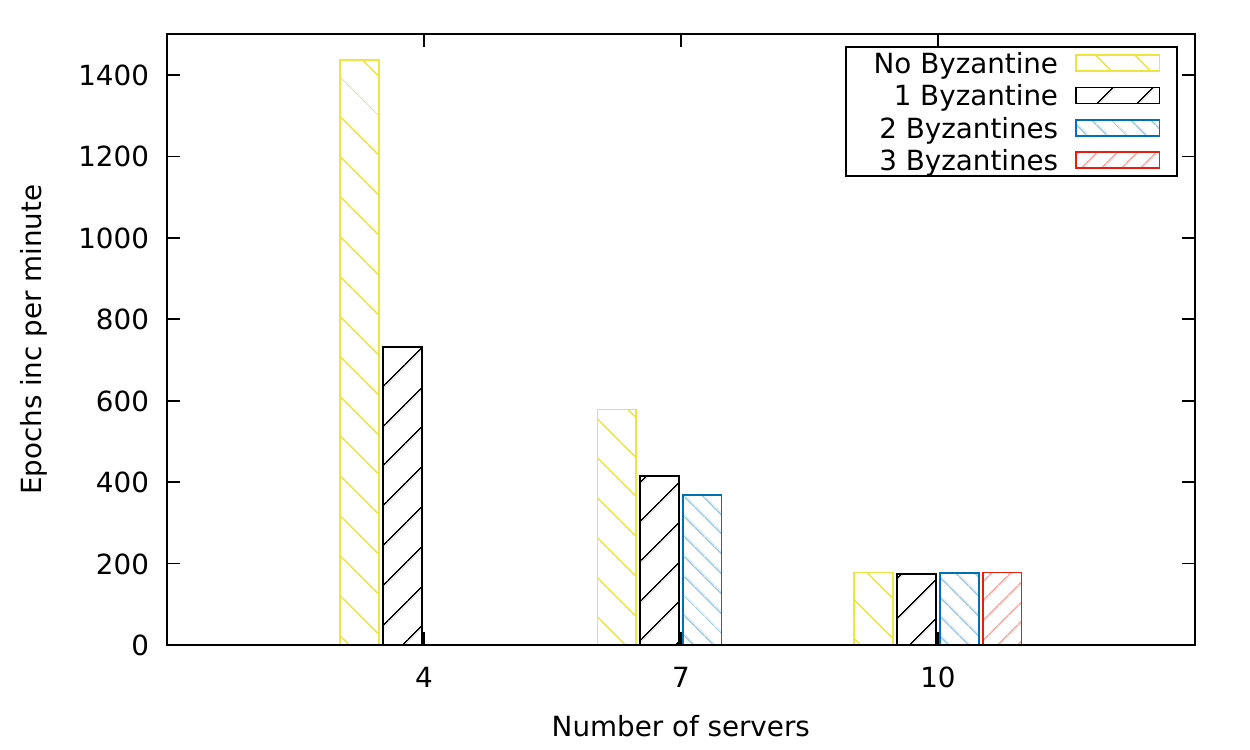} & \Exp{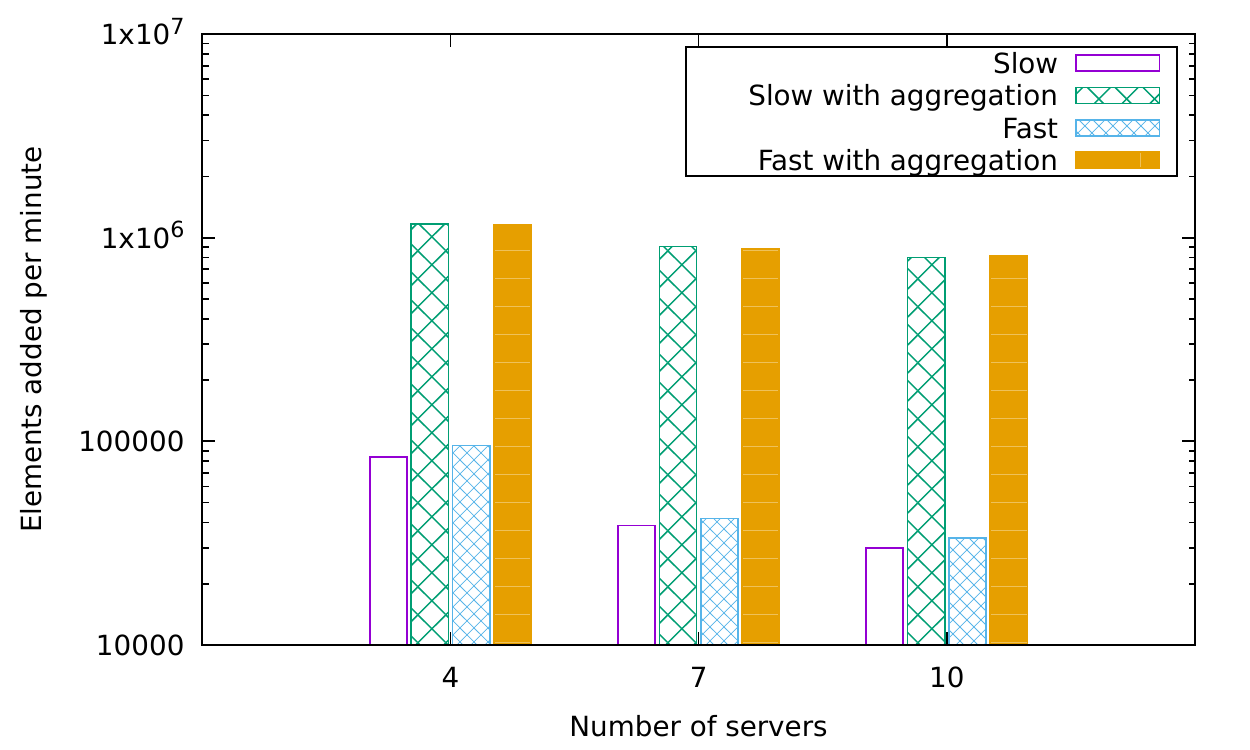} \\
     (a) Maximum epoch changes  & (b) Maximum elements added (no epochs)\\
     \Exp{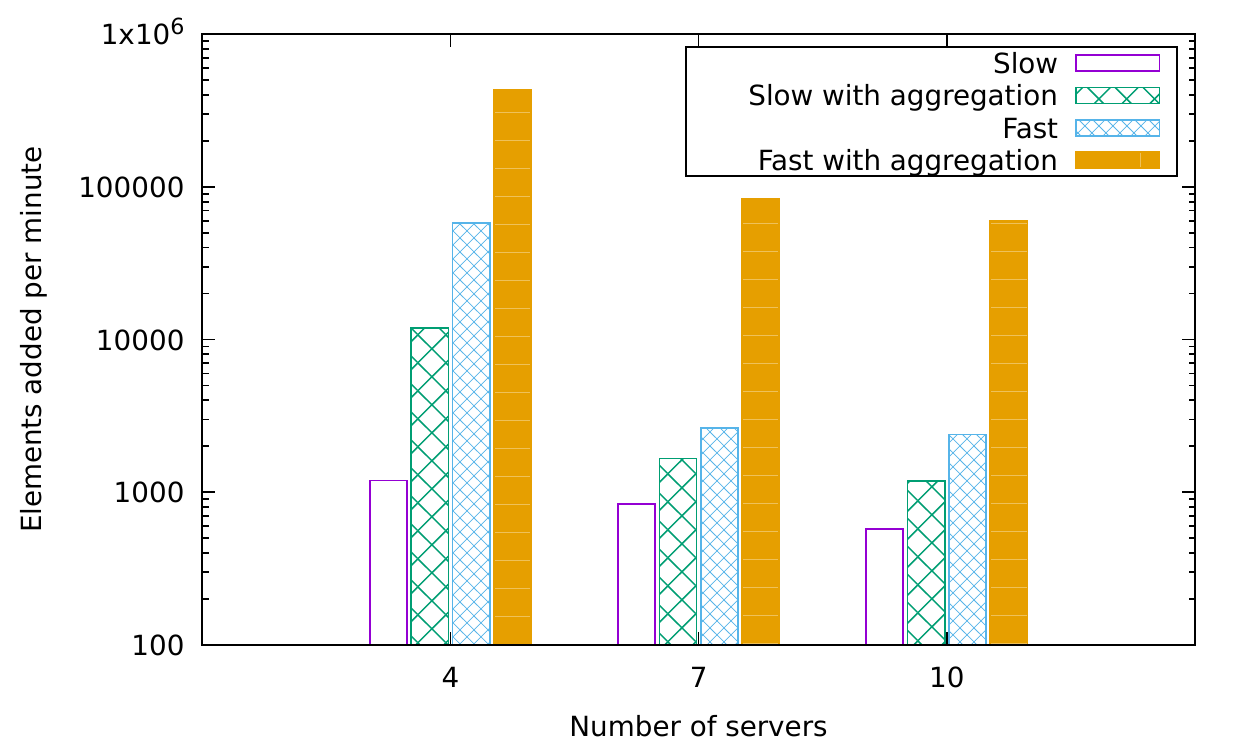} & \Exp{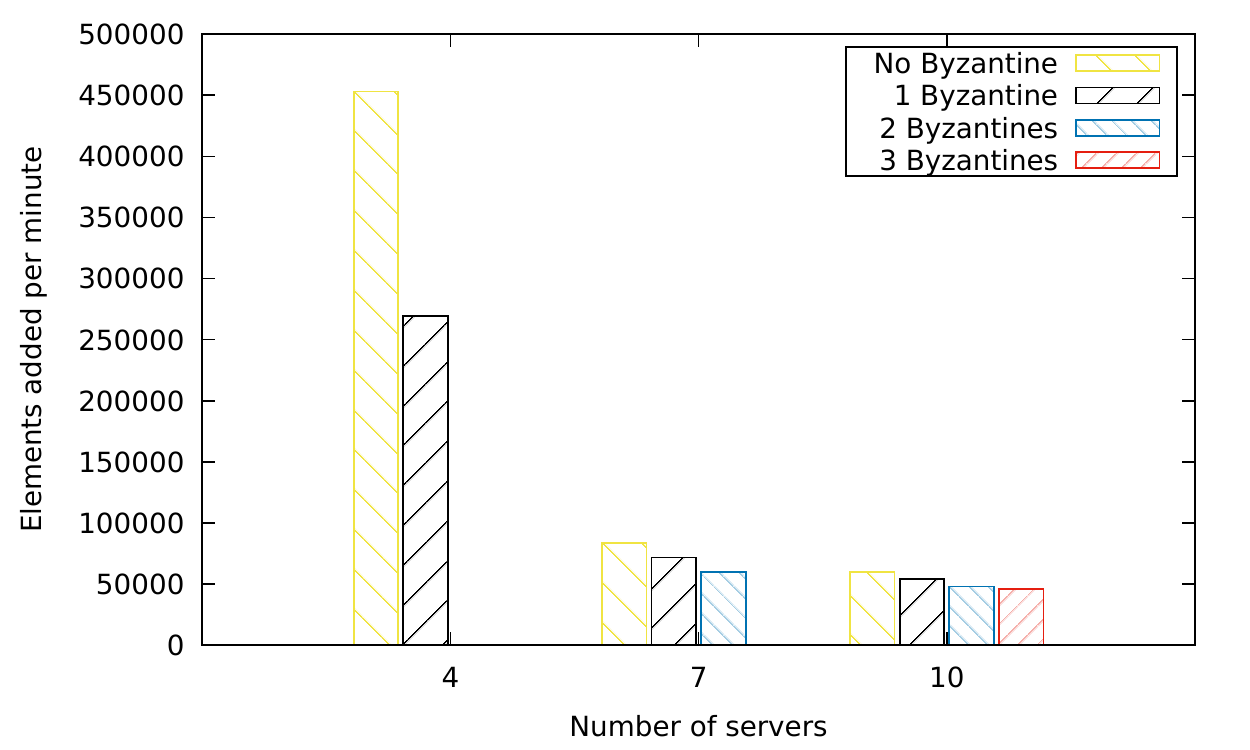} \\
     (c) Maximum elements added (with epochs) & (d) Silent Byzantine effect in adds (\AlgThreeSet)\\
     \Exp{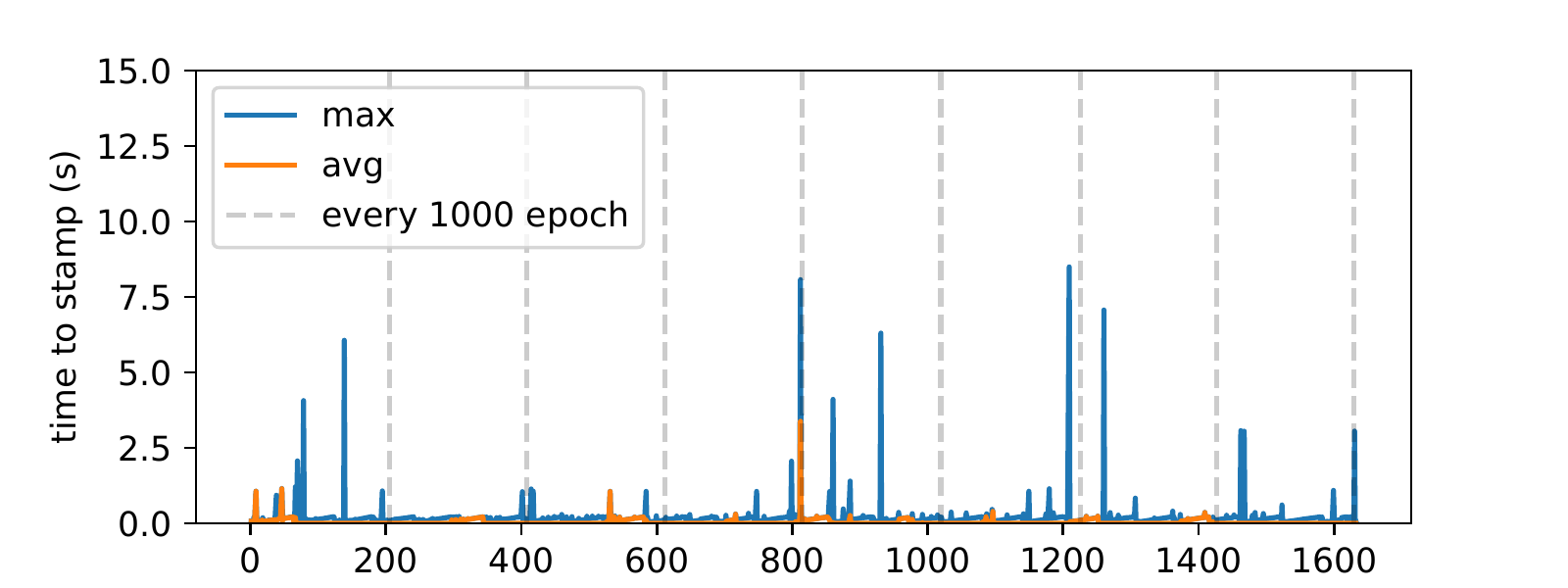} &     \Exp{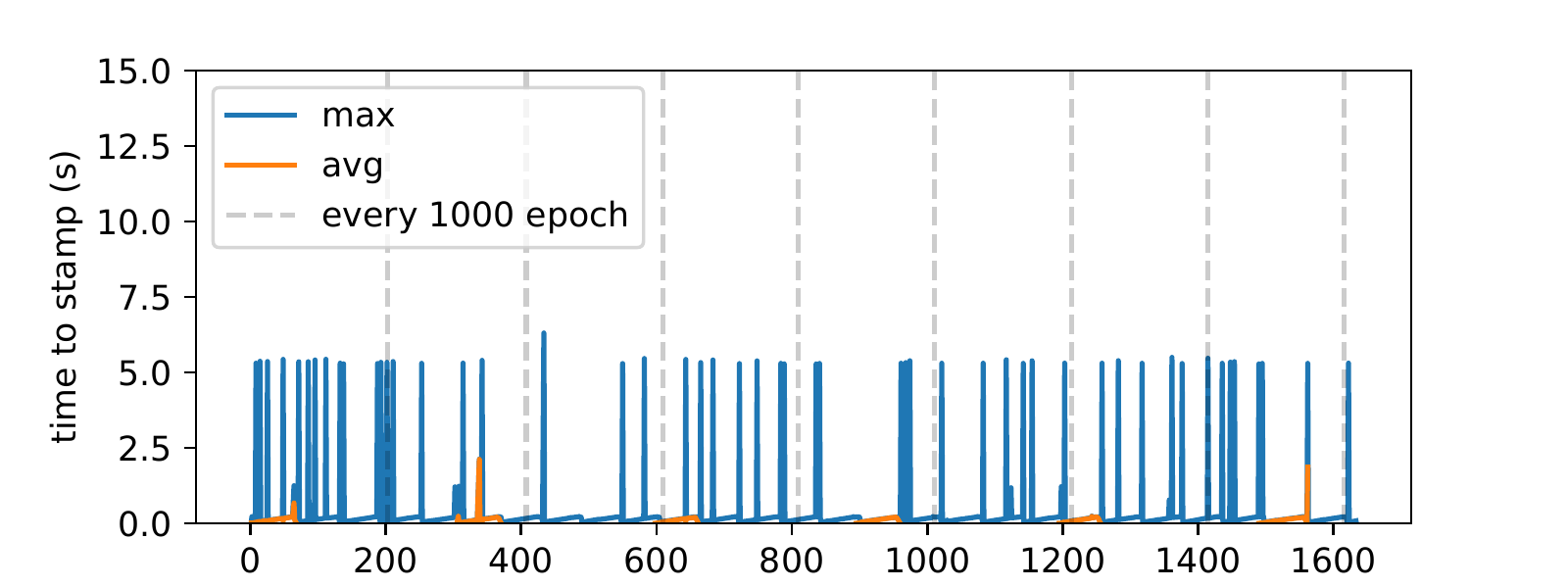} \\
     (e) Time to stamp (\AlgThree) & (f) Time to stamp (\AlgThreeSet)\\
   \end{tabular}
  \caption{Experimental results. \AlgTwoSet and \AlgThreeSet are the versions of the algorithms with aggregation. Byzantine servers are simply silent.}
  \label{fig:experiments}
\end{figure*}

%
%
%
We implemented the server code of \AlgTwo and \AlgThree, using our own
implementations of DSO, BRB and SBC.
%
%
Our prototype is written in Golang~\cite{donovan15go} 1.16 with
message passing using ZeroMQ~\cite{zeromq} over TCP.
Our testing platform used Docker running on a server with 2 Intel
Xeon CPU processors at 3GHz with 36 cores and 256GB RAM, running
Ubuntu 18.04 Linux64.
Each \setchain server was packed in a Docker container with no limit
on CPU or RAM usage.
\AlgTwo implements \setchain and DSO as two standalone executables
that communicate using remote procedure calls on the internal loopback
network interface of the Docker container.
The RPC server and client are taken from the Golang standard
library.
For \AlgThree everything resides in a single
executable.
For both algorithms, we evaluated two versions, one where
each element insertion causes a broadcast, and another where servers
aggregate locally the elements inserted until a maximum message size
(of $10^6$ elements) or a maximum element timeout (of 5s) is reached.
In all cases, elements have a size of 116-126 bytes.
Alg.~\ref{DPO-alg3-aggregated} implements the aggregated version of
\AlgThree.
  
\begin{figure}[t!]
	\begin{algorithm}[H]
		\renewcommand{\thealgorithm}{Fast with aggregation}         	
          \caption{\small Server implementation using a local set,
            BRB and SBC.
            }~\label{DPO-alg3-aggregated}
          \small
		\begin{algorithmic}[1]
			\State  \textbf{Init:} $\<epoch> \leftarrow 0$,
			\hspace{2em} $\<history> \leftarrow \emptyset$
			\State  \textbf{Init:} $\<theset> \leftarrow \emptyset$,
			\hspace{0.9em} $\<tobroadcast> \leftarrow \emptyset$
			\Function{\<Get>}{~}
			        \State \<return> $(\<theset>, \<history>,\<epoch>)$
			\EndFunction
			\Function{\<Add>}{$e$}
                       		 \State \<assert> $valid(e)$ and $e \notin \<theset>$
	                         \State $\<tobroadcast> \leftarrow \<tobroadcast> \cup \{e\}$ 
			\EndFunction	         		
			\Upon{\<BRB>.\<Deliver>(\<madd>\((s)\))}
                                \State \<assert> $valid(s)$
			        \State $\<theset> \leftarrow \<theset> \cup s$ 
  			        \State $\<tobroadcast> \leftarrow \<tobroadcast> \setminus s$ 
			\EndUpon
                        
                        %
                         
                        \State \ldots     \Comment{$\<EpochInc>$ and $\<BRB>.\<Deliver>(\<mepochinc>(h))$ as in Alg.~\ref{DPO-alg3}}
                        \setcounter{ALG@line}{23}
			\Upon{\<SBC>[\(h\)].\<SetDeliver>($propset$) and  $h==\<epoch>+1$}
			        \State $E \leftarrow \{e : e \in propset[j], valid(e) \wedge e\notin \<history>\}$
				\State $\<history> \leftarrow \<history> \cup \{\langle h, E \rangle\}$
				\State $\<theset> \leftarrow \<theset> \cup E$
                                \State $\<tobroadcast> \leftarrow \<tobroadcast> \setminus E$ 
			\EndUpon
                        \When{$|\<tobroadcast>| > 1000000$ or $\<tobroadcast>.oldest > 5s$}	
		 		\State \<BRB>.\<Broadcast>(add(\<tobroadcast>)) 
         			\State $\<tobroadcast> \leftarrow \emptyset$ 
	          	\EndWhen        
		\end{algorithmic}
	\end{algorithm}\vspace{-2em}
\end{figure}

%
%
We evaluated empirically the following hypothesis:
\begin{itemize}
\item (\Hspeed): The maximum rate of elements that can be inserted is
  much higher than the maximum epoch rate.
\item (\Hbetteralg): Alg.~\ref{DPO-alg-basb3} performs better than
  Alg.~\ref{DPO-alg-basb}.
\item (\Hagr): Aggregated versions perform better than their basic
  counterparts.
\item (\Hbyzantine): Silent Byzantine servers do not affect dramatically the
  performance.
\item (\Hdegrade): Performance does not degrade over time.
\end{itemize}

%
%

To evaluate hypotheses \Hspeed\ to \Hdegrade\, we carried out the
experiments described below, which are reported in
Fig.~\ref{fig:experiments}.
In all cases, operations are injected by clients running within the
same Docker container.
Resident memory was always enough such that in no experiment the
operating system needed to recur to disk swapping.
All experiments consider deployments with 4, 7, or 10 server nodes,
and each running experiment reported is taken from the average of 10
executions.

We tested first how many epochs per minute our \setchain
implementations can handle.
In these runs, we did not add any element and we incremented the epoch
rate to find out the smallest latency between an epoch and the
subsequent one.
We run it with 4, 7, and 10 nodes, with and without Byzantines
servers.
The outcome is reported in Fig.~\ref{fig:experiments}(a).
  
In our second experiment, we estimated empirically how many elements
per minute can be added using our four different implementations of
\setchain (\AlgTwo and \AlgThree with and without aggregation),
without any epoch increment.
This is reported in Fig.~\ref{fig:experiments}(b).
In this experiment \AlgTwo and \AlgThree perform similarly.
With aggregation \AlgTwo and \AlgThree also perform similarly, but one
order of magnitude better than without aggregation, confirming
(\Hagr).
Fig.~\ref{fig:experiments}(a) and (b) together suggest that sets are
three orders of magnitude faster than epoch changes, confirming
(\Hspeed).
  
The third experiment compares the performance of our implementations
combining epoch increments and insertion of elements.
We set the epoch rate at 1 epoch change per second and calculated the
maximum ratio of \<Add> operations.
The outcome is reported in
Fig.~\ref{fig:experiments}(c), which shows that \AlgThree outperforms
\AlgTwo.
In fact, \AlgThreeSet even outperforms \AlgTwoSet by a factor of
roughly 5 for 4 nodes and by a factor of roughly 2 for 7 and 10 nodes.
\AlgThreeSet can handle 8x the elements added by \AlgThree for 4 nodes
and 30x for 7 and 10 nodes.
The benefits of \AlgThreeSet over \AlgThree increase as the number of
nodes increase because \AlgThreeSet avoids broadcasting of elements
which generates a number of messages that is quadratic in the number
of nodes in the network.
This experiment confirms (\Hbetteralg) and
(\Hagr).
The difference between \AlgThree and \AlgTwo was not observable in the
previous experiment (without epoch changes) because the main
difference is in how servers proceed to collect elements to vote
during epoch changes.

The next experiment explores how silent Byzantine servers affect
\AlgThreeSet.  We implement silent Byzantine servers and run for 4, 7
and 10 nodes with an epoch change ratio of 1 epoch per second,
calculating the maximum add rate.
This is reported in~Fig.~\ref{fig:experiments}(d).
Silent Byzantine servers degrade the speed for $4$ nodes as in this
case the implementation considers the silent server very frequently in
the validation phase, but it can be observed that this effect is much
smaller for larger number of servers, validating (\Hbyzantine).
  
In the final experiment, we run 4 servers for a long time (30
  minutes) with an epoch ratio of 5 epochs per second and add
  requests to 50\% of the maximum rate.
  We compute the time elapsed between the moment in which clients
  request an add and the moment at which elements are stamped.
  Fig.~\ref{fig:experiments}(e) and (f) show the maximum and average
  times for elements inserted in the last second.
  In the case of \AlgThree, the worst case during the 30 minutes experiment
  was around 8 seconds, but the majority of elements were
  inserted within 1 sec or less.
  For \AlgThreeSet the maximum times were 5 seconds repeated in many
  occasions during the long run (5 seconds was the timeout to force a
  broadcast). This happens when an element fails to be inserted using
  the set consensus and ends up being broadcasted.
  In both cases, the behavior does not degrade with long runs,
  confirming (\Hdegrade).


  Considering that epoch changes are essentially set consensus, our
  experiments suggest that inserting elements in a \setchain is three
  orders of magnitude faster than performing consensus.
  However, a full validation of this hypothesis would require to fully
  implement \setchain on performant gossip protocols and compare with
  comparable consensus implementations.

\section{Client Protocols}
\label{sec:wrappers}

\begin{figure}[t!]
  \begin{algorithm}[H]
    \caption{\small Correct client protocol for DPO (for Alg.~\ref{DPO-alg-basb} and~\ref{DPO-alg-basb3}).}
    \label{DPO}
    \small
    \begin{algorithmic}[1]
      \Function{\<DPO>.\<Add>}{$e$}
      \State \<call> \<Add>\((e)\) in \(f+1\) different servers.
      \EndFunction
      \Function{\<DPO>.\<Get>}{~}
      \State \<call> \<Get>\(()\) in at least \(3f+1\) different servers.
      \State \<wait> \(2f+1\) responses $s.(\THESET, \HISTORY,\<epoch>)$ 
      \State $S \leftarrow \{ e | e\in{}s.\THESET \text{ in at least } f+1 \text{ servers } s \}$
      \State $H \leftarrow \emptyset$
      \State $i \leftarrow 1$
      \State $\Servers \leftarrow \{ s : s.\<epoch> \geq i \}$
      \While {$\exists E: |\{s \in \Servers: s.\<history>(i) = E \}| \geq f+1$}~\label{DPO-bwhile}
        \State $H \leftarrow H \cup \{ \langle i, E \rangle \} $
        \State $\Servers \leftarrow \Servers \setminus \{ s : s.\<history>(i)\neq E\}$ 
        \State $\Servers \leftarrow \Servers\setminus \{s: s.\<epoch>=i\}$         
        \State $i \leftarrow i+1$
      \EndWhile~\label{DPO-ewhile}
      \State \<return> $(S, H, i-1)$
      \EndFunction
		    \Function{\<DPO>.\<EpochInc>}{$h$}
			\State \<call> \<EpochInc>\((h)\) in \(f+1\) different servers.
			\EndFunction
    \end{algorithmic}
  \end{algorithm}
  \vspace{-2em}
\end{figure}

All properties and proofs in Section~\ref{sec:properties} consider the
case of clients contacting correct servers, and the implementations in
Section~\ref{sec:implementation} do not offer any guarantee to
clients about whether a server is correct.
Since client processes cannot know if they are contacting a Byzantine
or a correct server, a client protocol is required to guarantee
that they are exercising the interface of \setchain\ correctly.

In this section, we describe two such protocolos.
First, we present a client protocol inspired by the DSO clients
in~\cite{Cholvi2021BDSO}, which involves contacting several servers
per operation.
Later we present a more efficient ``optimistic'' solution, based on
try-and-check, that involves only a simple change in the servers.

\subsection{\setchain as a Distributed Partial Order Objects  (DPO)}

The general idea of the client protocol is to interact with enough
servers to guarantee that enough servers perform the desired
operation.
%
%
Alg.~\ref{DPO} shows the client protocol.
For $\<Add>$ and $\<EpochInc>$, to guarantee contacting at least one
correct server, we need to send $f+1$ requests to different servers.
%
%
Note that, depending on the protocol, each message may trigger
different broadcasts.

The wrapper algorithm for function $\<Get>$ can be split in two parts.
First, the client protocol contacts $3f+1$ nodes, and waits for at
least $2f+1$ responses ($f$ Byzantine servers may refuse to respond).
The response from server $s$ is $(s.\<theset>,s.\<history>,s.\<epoch>)$.
The protocol then computes $S$ as those elements known to be
in $\<theset>$ by at least $f+1$ servers (which includes at least one
correct server).
%
To compute $H$, the code proceeds incrementally epoch by epoch,
stopping at an epoch $i$ when less than $f+1$ servers agree on a set
$E$ of elements in $i$ (note that if $f+1$ servers agree that $E$ is
the set of elements in epoch $i$, this is indeed the case).
We also remove from $\Servers$ those servers that either do not know
the epoch $i$ (either slow processes or Byzantine servers) or that
incorrectly report something different than $E$ (a Byzantine server).
Once this process ends, the client returns $S$, $H$, and the latest
epoch computed.
It is guaranteed that $\<history>\subseteq\<theset>$.

\subsection{A Fast Optimistic Client}

We also present an alternative faster optimistic client.
In this approach correct servers also sign cryptographically, for
every epoch $i$, a hash of the set of elements in $i$, and insert this
hash in the \setchain\ as an element, see Alg.~\ref{server-sign}.
The number of additional elements added to the \setchain by the
servers is linear in the number of epochs, order of magnitudes smaller
than the number of element in the \setchain.

\begin{figure}[t!]
  \begin{algorithm}[H]
    \caption{\small Extend Alg. \ref{DPO-alg-basb3} adding the new epoch's hash cryptographically signed.}
    \label{server-sign}
    \small
    \begin{algorithmic}[1]
      \setcounter{ALG@line}{21} 
      \State \ldots     \Comment{previous lines as in Alg.~\ref{DPO-alg-basb3}}
      \Upon{\<SBC>[\(h\)].\<SetDeliver>($propset$) and  $h==\EPOCH+1$}
         \State $E \leftarrow \{e : e \in propset[j], valid(e) \wedge e\notin \HISTORY\}$
         \State $\HISTORY \leftarrow \HISTORY \cup \{\langle h, E \rangle\}$
	 \State $\THESET \leftarrow \THESET \cup E$
         \State $\EPOCH \leftarrow \EPOCH+1$
         \State $\<Add>(epoch\_signature(h,sign(hash(<<h,E>>))))$
      \EndUpon
    \end{algorithmic}
    \end{algorithm}
\end{figure}

Alg.~\ref{optimistic} shows the optimistic client described below.
\begin{figure}[t!]
  \begin{algorithm}[H]
    \caption{\small Optimistic client protocol for DPO (for Alg.~\ref{DPO-alg-basb} and~\ref{DPO-alg-basb3}).}
    \label{optimistic}
    \small
    \begin{algorithmic}[1]
      \Function{AddAndCheck}{$e$}
      \State \<call> \<Add>\((e)\) in 1 server.
      \State \<wait> $\Delta_g$
      \State \<call> \<Get>\(()\) in 1 server.
      \State \<wait> resp $(\<theset>, \<history>,\<epoch>)$
      \If {$\exists E,i: \<history>(i) = E  \wedge e \in E \wedge (h,hash(E))$ \\
        \hspace{2em} signed by $f+1$ different servers is in $\<theset>$}
          \State \<return> OK.
      \Else
          \State \<return> Fail.
      \EndIf
      \EndFunction
    \end{algorithmic}
  \end{algorithm}
  \vspace{-2em}
\end{figure}

To insert an element $e$ in the \setchain, the optimistic client
perform \emph{a single} \(\<Add>(e)\) request to one server, hoping it
will be a correct server.
After waiting for some time, the client invokes a $\<Get>$ from
\textbf{a single} server (which again can be correct or Byzantine) and
checks whether $e$ appears to be some epoch whose hash is
signed by (at least) $f+1$ different servers.
Note that simply receiving a $\<history>$ in which $e$ is in an epoch
is not enough to guarantee that $e$ has been added nor stamped,
because a Byzantine server can lie.
However, since cryptographic signatures cannot be forged,if $f+1$
servers sign the hash of an epoch (whose hash the client also 
computes) then at least one correct server certified the contents of
the epoch.
In comparison, the optimistic client can suffers a higher latency than
the DPO in Alg.~\ref{DPO}, because the optimistic client needs to wait
after adding an element to check that the element has been inserted
(or repeat).
On the other hand, this client requires only one message per $\<Add>$
and one message per $\<Get>$ in the usual case that correct servers
are reached, which reduces dramatically the number of messages
interchanged.
Note that each $\<Add>$ request triggers a BRB.Broadcast producing a
cascade of messages that is quadratic on the number of servers.
Hence, these two clients present a tradeoff between latency and
throughput.


\section{Modelling Byzantine Behavior}\label{sec:byzantine}

Formally proving properties\footnote{We refer to rigorous machine reproducible
or checkable proofs.} of Byzantine tolerant distributed algorithms is a very
challenging task.
In this section, we introduce a non-deterministic process (see
Alg.~\ref{DPO-alg3-byz}) that abstracts the combined behavior of all
Byzantine processes (Alg.~\ref{DPO-alg3}).
This technique reduces a scenario with $n-f$ correct servers and $f$
Byzantine servers into an \emph{equivalent} scenario with $n-f$ correct
servers and one non-deterministic server, allowing to leverage many
recent techniques for formally proving properties of (non-Byzantine)
distributed algorithms.
Our non-deterministic process abstracts Byzantine behaviour, even for
Byzantine processes that enjoy instantaneous communication among
themselves to perform coordinated attacks.

\setchain\ assumes that Byzantine processes cannot forge
valid elements~(Section~\ref{sec:model:computation}), they only become
aware of their existence when correct servers send these elements or
when they are inserted by clients.
We assume that as soon as a Byzantine process receives a valid element
$e$ all other Byzantine process know $e$ too.

%
We extend all algorithms with a new primitive,
$\<SBC>[h].\<Inform>(\<prop>)$, that is triggered when servers call
$\<SBC>[h].\<Propose>(\<prop>)$ and satisfies the following property:
\begin{itemize}
\item\textbf{SBC-Inform-Validity}: if a process executes
  $\<SBC>[h].\<Inform>(\<prop>)$ then some other process executed
  $\<SBC>[h].\<Propose>(\<prop>)$ in the past.
\end{itemize}
We use this property to equip Byzantine servers with the ability to discover
elements proposed by correct servers during set consensus, but that are not
assigned an epoch yet.
Note that Byzantine servers can return elements that were proposed but not
assigned to an epoch in a response to a \(\<Get>()\) invocation.
However, we would not know where those elements came from if we do not include
primitive \(\<SBC>[h].\<Inform>\).
%

We assume there are a functions
\textit{havoc}\_\textit{subset},
\textit{havoc}\_\textit{partition},
\textit{havoc}\_\textit{element},
\textit{havoc}\_\textit{number},
and
\textit{havoc}\_\textit{invalid}\_\textit{elems},
that generate, respectively, a random subset, element and partition
from a given set; a random number and random invalid elements.
We do not focus on the semantics of these functions, we use them to
model Byzantine processes producing arbitrary set of elements taken
from a set of known values.

We model all Byzantine processes combined behaviour in Alg.~\ref{DPO-alg3-byz}, where it
maintains a local set $\Knowledge$ to record all valid elements that the
``Byzantine'' server is aware of.
It provides the same interface as Alg.~\ref{DPO-alg3}, with an
additional function \<DoStuff>.
The function \<DoStuff> is invoked when the process starts and
non-deterministically emits messages at arbitrary times
(lines~\ref{alg3byz-dostuff-begin}-\ref{alg3byz-dostuff-end}), using BRB and SBC
primitives as these are the only messages that may cause any effect on correct
servers.
Note that these messages can have invalid elements or valid elements
known to the non-deterministic process.
Similarly, when clients invoke \(\<Get>()\), it returns
$(s_{v}\cup s_{i}, partition(s'_{v}\cup s'_{i}),h)$, where $s_{v}$ and
$s'_{v}$ are sets of valid elements from $\Knowledge$ while $s_{i}$ and
$s'_{i}$ are sets of invalid elements
(lines~\ref{alg3byz-get-begin}-\ref{alg3byz-get-end}).
Upon receiving a message, the non-deterministic process annotates all
newly discovered valid elements in its local set \Knowledge.

We show that any execution of \setchain maintained by $n$ servers
implementing Alg.~\ref{DPO-alg3} out of which at most $1 \leq f < n/3$
are Byzantine can be mapped to an execution of \setchain maintained by
$n-f$ \textbf{correct} servers implementing Alg.~\ref{DPO-alg3} and
\textbf{one} server implementing Alg.~\ref{DPO-alg3-byz}, and vice
versa.

\begin{figure}[t!]
	\begin{algorithm}[H]
          \caption{\small Capture Byzantine Behaviour
            }~\label{DPO-alg3-byz}
          \small
		\begin{algorithmic}[1]
		        \State  \textbf{Init:} $\<knowledge> \leftarrow \emptyset$
			\Function{\<Get>}{~} \label{alg3byz-get-begin}
            \State \<return> $(\HavocSubset(\<knowledge> \cup \GenerateInvalidElems()),$\\
          \hspace{5em}  $\HavocPartition(\HavocSubset(\<knowledge>\cup \GenerateInvalidElems())),\HavocNumber())$
			\EndFunction \label{alg3byz-get-end}
			\Function{\<Add>}{$e$} \label{alg3byz-add-begin}
                                \State \<assert> $valid(e)$
			        \State $\<knowledge> \leftarrow \<knowledge> \cup \{e \}$  
                        \EndFunction \label{alg3byz-add-end}
			\Upon{\<BRB>.\<Deliver>(\<madd>\((e)\))} \label{alg3byz-brbdeliver-add-begin}
                                \State \<assert> $valid(e)$ 
                                \State $\<knowledge> \leftarrow \<knowledge> \cup \{e \}$ 
			\EndUpon \label{alg3byz-brbdeliver-add-end}
			\Function{\<EpochInc>}{$h$} \label{alg3byz-epochinc-begin}
				\State \textbf{return}
                        \EndFunction \label{alg3byz-epochinc-end}
			\Upon{\<BRB>.\<Deliver>(\<mepochinc>($h$))} \label{alg3byz-brbdeliver-epinc-begin}
			        \State \textbf{nothing}
			\EndUpon \label{alg3byz-brbdeliver-epinc-end}
			\Upon{\<SBC>[\(h\)].\<SetDeliver>($propset$)} \label{alg3byz-setdeliver-begin}
				\State $\<knowledge> \leftarrow \<knowledge> \cup \{e : e \in propset \wedge valid(e) \}$ 
			\EndUpon \label{alg3byz-setdeliver-end}
			\Upon{\<SBC>[\(h\)].Inform($prop$)} \label{alg3byz-sbcinform-begin}
			        \State $\<knowledge> \leftarrow \<knowledge> \cup \{e : e \in prop \wedge valid(e) \}$
			\EndUpon \label{alg3byz-sbcinform-end}
			\Function{\<DoStuff>}{}
                        \While{true} \label{alg3byz-while-begin}
				        \State\label{alg3byz-brbbroadcast-add}\label{alg3byz-dostuff-begin}
                            $\<BRB>.\<Broadcast>(add(\HavocElement(\<knowledge>\cup \GenerateInvalidElems())))$
                        \State \(||\) 
	                    \State \label{alg3byz-brbbroadcast-epinc}
                            $\<BRB>.\<Broadcast>(epinc(\HavocNumber()))$
                        \State \(||\) 
                        \State\label{alg3byz-sbcpropose}
                              $\<SBC>[\HavocNumber()].\<Propose>(\HavocSubset(\<knowledge>\cup
                              \GenerateInvalidElems()))$
                        \State \(||\) 
	                    \State \textbf{nothing}
                        \EndWhile\label{alg3byz-while-end}\label{alg3byz-dostuff-end}
			\EndFunction
		\end{algorithmic}
	\end{algorithm}\vspace{-2em}
\end{figure}

We denote with $\Gamma$ the \emph{model} that represents the execution
of a \setchain\ that
is maintained by $n$ processes implementing Alg.~\ref{DPO-alg3} out which
$1\leq f < n/3$ are Byzantine servers.

We denote with $\correctset \subset \mathbb{D}$ the set of $n-f$
correct servers and with
$\{b_0, \ldots, b_{f-1}\} = \byzantineset \subset \mathbb{D} \setminus
\correctset$ the set of Byzantine servers.


\paragraph{Events}\label{para:events}
We represent with events the different interactions users and servers can have
with a \setchain\ plus the internal event of the \setchain\ reaching consensus.
\begin{compactitem}
\item $\<get>()$ represents the invocation of function $\<Get>$,
\item $\<add>(e)$ represents the invocation of function
  $\<Add>(e)$,
\item $\<BRB>.\<Broadcast>(x)$ represents the broadcast of \(\<madd>\) or
  \(\<mepochinc>\) messages through the network, 
  \(x = \<madd>(e)\) or \(x =\<mepochinc>(h)\) respectively,
\item $\<BRB>.\<Deliver>(x)$ represents the reception of
  message $x$,
\item
  $\<EpochInc>(h)$ represents the invocation of function
  $\<EpochInc>(h)$,
\item
  $\<SBC>[h].\<Propose>(\<prop>)$ represents the propose of \(\<prop>\) in
  the $h$ instance of SBC,
\item $\<SBC>[h].\<Inform>(\<prop>)$ represents the reception of
  $\<SBC>[h].\<Propose>(\<prop>)$, 
\item \(\<SBC>[h].\<SetDeliver>(\<propset>)\) represents that
  \(\<propset>\) is the result of the \(h\) instance of SBC,
\item  $\<SBC>[h].\<Consensus>(\<propset>)$ is the internal event that denotes that
  consensus for epoch $h$ is reached,
\item $\<nop>$ represents an event where nothing happens.
\end{compactitem}

The event $\<Consensus>(\<propset>)$ models the point in time at which set
consensus is final and all correct servers will agree on the set $\<propset>$.
%

All events \(\EV\), except $\<Consensus>$ and $\<nop>$, happen in a particular
server $s$ therefore they have an attribute \(\EV.\<server>\) that returns
$s$.

\paragraph{Network}
We model the network as a map from servers to tuples of the form (\<sent>,
\<pending>, \<received>).
For a server \(s\), \(\Delta(s).\<sent>\) is the sequence of messages
sent from server $s$ to other servers, $\Delta(s).\<pending>$ is a multiset that
contains all messages sent from other servers that $s$ has not processed yet
and $\Delta(s).\<received>$ is the sequence of messages received by server $s$.
The state of the network is modified when servers send or receive
messages.
When $s$ sends a message $m$ using \(\<BRB>.\<Broadcast>(m)\),
$m$ is added to the \<sent> sequence of \(s\) and to
the \<pending> multiset of all other servers.
Similarly, when a server proposes $\<prop>$ in the $h$ instance of SBC, the
message \(m = (h,prop)\) is added to its \<sent>
sequence and to the \<pending> multiset of
all servers:

\[
  \<send>(\Delta, m, s)(s') = \begin{dcases}
  (\Delta(s).\<sent> ++ m, \Delta(s).\<received>,
  \Delta(s).\<pending> \cup \{m\}) & \text{if } s = s' \\
  (\Delta(s').\<sent>, \Delta(s').\<received>, \Delta(s').\<pending> \cup \{m\}) & \text{otherwise}
                            \end{dcases}
\]

When servers receive a message $m$, $m$ is removed from their
\<pending> multiset and is inserted in their \<received> sequence:
\[
  \<receive>(\Delta, m,s)(s') = \begin{dcases}
  (\Delta(s).\<sent>, \Delta(s).\<received> ++ m,
  \Delta(s).\<pending> \setminus \{m\}) & \text{if } s = s' \\
  \Delta(s')                           & \text{otherwise}
                          \end{dcases}
\]
A \emph{configuration} $\Phi = (\Sigma,\Delta,H,K)$ for model $\Gamma$
consists of: a state $\Sigma$ mapping from correct servers to their
local state; a network $\Delta$ containing messages sent between
servers; a partial map $H$ from epoch numbers to sets of elements (the
history of consensus reached), and a set of valid elements $K$ that
have been disclosed to some Byzantine server (that is, we assume that
Byzantine servers may have instantaneous communication).

The initial configuration $\Phi_0 = (\Sigma_0,\Delta_0,H_0,K_0)$ is such that
$\Sigma_0(s)$ is the initial state of every correct process,
$\Delta_0$ is the empty network, $H_0$ is the empty map and
$K_0$ is the empty set.

We define whether an event $\EV$ is \emph{enabled} in
configuration $(\Sigma,\Delta,H,K)$ as follows:
\begin{itemize}
\item $\<get>()$ and $\<nop>$ are always enabled,
\item $\<add>(e)$ is enabled if $e$ is valid and either \(s =
  \<add>(e).\<server>\) is a Byzantine server or $e\notin
  \Sigma(s).S$, 
\item $\<BRB>.\<Broadcast>(x)$ is enabled if
  $\<BRB>.\<Broadcast>(x).\<server>$ is a Byzantine server and either
  \(x =\<mepochinc>(h)\) or \(x = \<madd>(e)\) with \(e \in K\) or $e$
  invalid,  
\item $\<BRB>.\<Deliver>(\<add>(e))$ is enabled if
  $\<add>(e) \in \Delta(s).\<pending>$ and $e$ is valid,
\item $\<EpochInc>(h)$ is enabled if either \(s =
  \<EpochInc>(h).\<server>\) is a Byzantine server or
  $h = \Sigma(s).\<epoch>+1$,
\item $\<BRB>.\<Deliver>(\<mepochinc>(h))$  is enabled if
$\<mepochinc>(h) \in \Delta(s).\<pending>$ and
either \(s = \EV.\<server>\) is a
Byzantine server, $h < \Sigma(s).\<epoch> + 1$ or $h =
\Sigma(s).\<epoch>+1$ plus $s$ has not propose anything for the $h$
instance of SBC: \(\nexists \<prop> : (h,prop) \in \Delta(s).sent\), 
\item \(\<SBC>[h].\<Propose>(\<propset>)\) is enabled if
  \(\<SBC>[h].\<Propose>(\<propset>).\<server>\) is a Byzantine server
  and all valid elements in $\<propset>$ are known by Byzantine
  servers: $\{ e \in \<propset> : valid(e) \} \subseteq K$,
\item $\<SBC>[h].\<Inform>(\<prop>)$ is enabled if
  $(h,prop) \in \Delta(s).\<pending>$,  
\item $\<SBC>[h].\<SetDeliver>(\<propset>)$ is enabled if $H(h) =
  \<propset>$ and either \(s = \EV.\<server>\) is
  a Byzantine server or $h=\Sigma(s).\<epoch>+1$,
\item $\<SBC>[h].\<Consensus>(\<propset>)$ is enabled if $H(h-1)$ is defined,
  $H(h)$ is undefined, at least
one process proposed a set before (
\(\exists s, h, \<prop>:  (h, \<prop>) \in \Delta(s).\<sent>\)), and
\(\<propset>\) is a subset of the union of all elements proposed for the
$h$ instance of SBC
($\<propset> \subseteq \bigcup_{r} \{p : (h,p) \in \Delta(r).\<sent>\}$), 
\end{itemize}

Then we defined the \emph{effect} of an event $\EV$ on a
configuration $(\Sigma,\Delta,H,K)$, where $\EV$ is enabled in
$(\Sigma,\Delta,H,K)$ as the following configuration
$(\Sigma',\Delta',H',K')$.
\begin{itemize}
  \item For the set $K'$, if event \(\EV.\<server>\) is a Byzantine server
    then $K'=K\cup\ValidElements(\EV)$, otherwise $K'=K$.
    Where   \[
    \ValidElements(\EV) = \left\{ \begin{array}{lr}
  \{e\}, & \text{for }  \EV = \<add>(e) \wedge valid(e) \\
  \{e\}, &\text{for }  \EV = \<BRB>.\<Deliver>(\<add>(e)) 
          \wedge valid(e)  \\
  S, & \text{for } \EV = \<SBC>[h].\<Inform>(\<prop>) 
       \wedge S = \{e \in prop: valid(e) \} \\
  S, & \text{for } \EV = \<SBC>[h].\<SetDeliver>(\<propset>) 
       \wedge S = \{e \in propset:
  valid(e) \} \\
  \emptyset, & \text{otherwise } 
        \end{array}\right.
 \]
  \item For the network $\Delta'$,
 \begin{itemize}
  \item If event \(\EV\) is $\<add>(e)$ and \(s=\EV.\<server>\) is a correct server then
    $\Delta' = \<send>(\Delta, \<madd>(e), s)$;    
  \item if event \( \EV \) is $\<BRB>.\<Deliver>(\<add>(e))$ then
    $\Delta' = \<receive>(\Delta,\<add>(e),\EV.\<server>)$;  
  \item If event \(\EV\) is $\<EpochInc>(h)$ and \(s = \EV.\<server>\) is a correct server then\\
$\Delta' = \<send>(\Delta, \<mepochinc>(h), s)$;
  \item If event \(\EV\) is $\<BRB>.\<Deliver>(\<mepochinc>(h))$ and
\(s = \EV.\<server>\) is a Byzantine server or  $h < \Sigma(s).\<epoch> + 1
$ then $\Delta' = \<receive>(\Delta,\<mepochinc>(h),s)$
   \item If event \(\EV\) is $\<BRB>.\<Deliver>(\<mepochinc>(h))$ and
\(s = \EV.\<server>\) is a correct server and $h = \Sigma(s).\<epoch> +
1$ then let \(ps\) be the set of elements in \(s\) without an epoch,
\(ps = \Sigma(s).S \setminus \bigcup_{k}^{h-1}\Sigma(s).H(k)\), in the
new network  $\Delta' = \<send>(\<receive>(\Delta,\<mepochinc>(h),s),
(h,ps),s)$ 
  \item If event \(\EV\) is \( \<BRB>.\<Broadcast>(m)\) then $\Delta'
    = \<send>(\Delta, m,\EV.\<server>)$; 
  \item If event \(\EV\) is \(\<SBC>[h].\<Propose>(prop)\)  then
    $\Delta' = \<send>(\Delta, (h,prop),\EV.\<server>)$;
  \item If event \(\EV\) is \(\<SBC>[h].\<Inform>(prop)\)  then
    $\Delta' = \<receive>(\Delta, (h,prop),\EV.\<server>)$;  
  \item otherwise  $\Delta'=\Delta$.
  \end{itemize}
\item For the state map $\Sigma'$:
\begin{itemize}
\item If event \(\EV.\<server>\) is Byzantine then $\Sigma'=\Sigma$.
\item If event \(s = \EV.\<server>\) is a correct server then
  \begin{itemize}
  \item $\EV = \<BRB>.\<Deliver>(\<add>(e))$, then
    $\Sigma' = \Sigma \oplus \{s \mapsto (\Sigma(s).S \cup \{x : x =
    e \wedge \<valid>(x) \}, \Sigma(s).H,$\\$ \Sigma(s).\<epoch>)\}$.
  \item $\EV = \<SBC>[h].\<SetDeliver>(\<propset>)$ then
    $\Sigma' = \Sigma \oplus \{s \mapsto (\Sigma(s).S \cup E,
    \Sigma(s).H \cup \{\langle h,E \rangle \}, h)\}$ with
    $E = \{e : e \in \<propset>, \<valid>(e) \wedge e \notin
    \Sigma(s).H\}$
    \item otherwise $\Sigma'=\Sigma$.
  \end{itemize}
\end{itemize}
\item For the epoch map $H'$:
\begin{itemize}
  \item If the event is $\<SBC>[h].\<Consensus>(\<propset>)$, then
  $H'(h)=\<propset>$  and $H'(x)=H(x)$ for $x\neq h$.
   \item otherwise $H'=H$.
\end{itemize}
\end{itemize}

If event $\EV$ is enabled at configuration $(\Sigma,\Delta,H,K)$ and
$(\Sigma',\Delta',H',K')$ is the resulting configuration after
applying the effect of $\EV$ to $(\Sigma,\Delta,H,K)$ then we
write $(\Sigma,\Delta,H,K)\xrightarrow{\EV}(\Sigma',\Delta',H',K')$.

\begin{definition}[Valid Trace of $\Gamma$]
  A valid trace of $\Gamma$ is an infinite sequence
  $(\Sigma_0\Delta_0,H_0,K_0)\xrightarrow{\EV_0}
  (\Sigma_1,\Delta_1,H_1,K_1)\xrightarrow{\EV_1}\ldots$ such that
  $(\Sigma_0,\Delta_0,H_0,K_0)$ is the initial configuration.
\end{definition}

\begin{definition}[Reachable Configuration in $\Gamma$]
  A configuration $\Phi$ is \emph{reachable} in $\Gamma$ if it is the initial
  configuration, $\Phi = \Phi_0$, or it is the result of applying an
  enabled event in a reachable configuration: $\exists \Psi, \EV:
  \Psi \text{ is reachable } \wedge \Psi \xrightarrow{\EV} \Phi$.
\end{definition}

%
We denote with $\Gamma'$ the \emph{model} that represents the execution
of a \setchain\ that is maintained by
$n-f$ correct servers implementing Alg.~\ref{DPO-alg3} and
one server $b$ implementing Alg.~\ref{DPO-alg3-byz}.
A configuration in model \(\Gamma'\) is a tuple
$(\Sigma,\Delta,H,\Tau)$ where $\Tau$ is the local state of
server $b$, and $\Sigma$, $\Delta$ and $H$ are as in model \(\Gamma\).
The local state of server \(b\) consists in storing the knowledge
harnessed by all Byzantine processes in model \(\Gamma\).

The initial configuration $\Phi'_0 = (\Sigma_0,\Delta_0,H_0, \Tau_0)$ is such that
$\Sigma_0(s)$ is the initial state of every correct process, $\Tau_0$
is the empty set, $\Delta_0$ is the empty network and $H_0$ is the
empty map.

We define a valid trace in \(\Gamma'\) following the same idea as in model
\(\Gamma\).
For correct processes, we have the same rules as in the previous model.
The rules for process $b$ are similar to the ones for Byzantines
processes in the previous model.
The difference is that here when process $b$ consumes an event it stores all valid
elements contained in the event in its local state $\Tau$, to be able to later
use valid elements to produce events non-deterministically.

The effect in $\Gamma'$ of an event $\EV$ is now defined as follows.
Given a configuration $(\Sigma,\Delta,H,\Tau)$, where $\EV$ is
enabled in $(\Sigma,\Delta,H, \Tau)$, the effect of $\EV$ is a
configuration $(\Sigma',\Delta',H', \Tau')$ such that:
\begin{itemize}
\item For $\Delta'$, $\Sigma'$ and $H'$ the effect is as in $\Gamma$.
\item For $\Tau'$:
\begin{itemize}
\item If event \(\EV.\<server> = b\) then
  $\Tau' = \Tau \cup \ValidElements(\EV)$.
  \item otherwise $\Tau = \Tau'$
\end{itemize}
\end{itemize}

Note that the only ``Byzantine'' process now only annotates all events that
it discovers.
The definition of valid trace and reachable configuration are analogous as for $\Gamma$.

\begin{definition}[Valid Trace of $\Gamma'$]
  A valid trace of $\Gamma'$ is an infinite sequence
  $(\Sigma_0,\Delta_0,H_0,\Tau_0)\xrightarrow{\EV_0}
  (\Sigma_1,\Delta_1,H_1,\Tau_1)\xrightarrow{\EV_1}\ldots$ such that
  $(\Sigma_0,\Delta_0,H_0,\Tau_0)$ is the initial configuration.
\end{definition}  

Where the main difference with trace \(\Gamma\)-valid is that now we have a
definition for what a ``Byzantine'' server could do Alg.~\ref{DPO-alg3-byz}, and
we update its state accordingly.

\begin{definition}[Reachable Configuration in $\Gamma'$]
  A configuration $\Phi'$ is reachable in $\Gamma'$ if it is the initial
  configuration $\Phi' = \Phi'_0$ or it is the result of applying an
  enabled event in a reachable configuration: $\exists \Psi', \EV:
  \Psi' \text{ is reachable } \wedge \Psi' \xrightarrow{\EV} \Phi'$.
\end{definition}

Given that when server $b$ receives a message its valid elements are
added to its local state $\Tau$, by induction it follows that in a reachable
configuration all valid elements that appear in the \(\<received>\)
multiset of server $b$ are also in $\Tau$:   

\begin{lemma}\label{l:receive:tau}  
  Let \((\Sigma,\Delta,H,\Tau)\) be a reachable configuration in
  $\Gamma'$ then: \[\bigcup_{\EV \in
  \Delta(b).\<received>}\ValidElements(\EV) \subseteq \Tau\]
\end{lemma}

We show that model \(\Gamma\) and \(\Gamma'\) are observational
equivalent, in other words, that external users cannot distinguish
both models under the assumption that Byzantine processes may share
information.
In order to prove it, we show that for each valid trace in one model
there is a valid trace in the other model such that corresponding
configurations are indistinguishable. 
Intuitively, two configurations $\Phi = (\Sigma,\Delta,H,K)$ and
$\Phi' = (\Sigma',\Delta',H', \Tau')$ are observational equivalent,
\(\Phi \sim \Phi'\), if and only if (1) every correct process has the
same local state (2) the network are observational equivalent (3)
$\Tau$ and $K$ contain the same elements (4) the histories of consensus
reached are the same. 
%

The main idea is that since Byzantine processes shared their knowledge
outside the network, we can replace them by one non-deterministic
process capable of taking every possible action that Byzantine processes
can take.
Hence, our definition of observational equivalent are based on mapping every
Byzantine action in model \(\Gamma\) to possible an action of process \(b\) in model
\(\Gamma'\) and mapping every single \(b\) action in model \(\Gamma'\) to a sequence of
actions for the processes in model \(\Gamma\) (one for each Byzantine process).

\begin{definition}[Network Observational
  Equivalence]\label{def:network:obs:eq}
  Let \(\Delta\) be a network of model \(\Gamma\) and \(\Omega\) be a
  network of model \(\Gamma'\).
  We say that \(\Delta\) and \(\Omega\) are observational equivalent,
  \(\Delta ~ \Omega \), if and only if
  \begin{enumerate}
  \item they have the exact same state for correct processes: \(\forall
    s \in \correctset: \Delta(s) = \Omega(s)\);
  \item a message is sent by process $b$  if and only if it
    is sent by a Byzantine process: \\
    $\bigcup_{s \in\byzantineset} multiset(\Delta(s).\<sent>) =
    multiset(\Omega(b).\<sent>)$;
  \item all pending messages for process $b$ are pending in all
    Byzantine processes:\\  $\forall s\in\byzantineset: \Omega(b).\<pending>
    \subseteq \Delta(s).\<pending>$;
  \item a message is received by process $b$ if and only if
    it is received by at least one Byzantine process:  $multiset(\Omega(b).\<received>) \subseteq \bigcup_{s
    \in\byzantineset} multiset(\Delta(s).\<received>)$;
  \item all messages sent to a Byzantine process (messages either in
    \<pending> or in \<received>) where also sent to $b$: $\forall s \in\byzantineset: multiset(\Delta(s).\<received>)
    \cup \Delta(s).\<pending> =
    multiset(\Omega(b).\<received>) 
    \cup  \Omega(b).\<pending>$;
  \item a Byzantine process received a message if it was received
    by $b$:\\ $\forall s \in\byzantineset: multiset(\Delta(s).\<received>) \subseteq
    multiset(\Omega(b).\<received>)$.
    \end{enumerate}
\end{definition}

The following lemmas show that the functions \(\<send>\) and
\(\<receive>\) preserve observational equivalence relation between
networks.

\begin{lemma}\label{l:net:obs:eq:correct}
  Let \(\Delta\) be a network of model \(\Gamma\), \(\Delta'\) be a
  network of model \(\Gamma'\) and $s$ a correct server. Then,
  \begin{itemize}
    \item \(\<send>(\Delta,m,s) ~ (\<send>(\Delta',m,s) \),
    \item \(\<receive>(\Delta,m,s) ~ (\<receive>(\Delta',m,s) \).
  \end{itemize}
\end{lemma}

\begin{proof}
  First, we will prove that if a correct server send the same message in two
  observational equivalent networks then the networks remain observational
  equivalents.
  
  Let $s_0$ be a correct server, then the observational equivalence
  between \(\Delta\) and \(\Delta'\) implies that \(\Delta(s_0) =
  \Delta'(s_0)\).
  As function \(\<send>\) only modifies a network by adding a message $m$
  to the \(\<pending>\) component of all servers and to \(\<sent>\) component
  of the sender (in both cases the same server $s$) it follows that 
  (1) \(\forall s_0\in\correctset: \<send>(\Delta,m,s)(s_0) = \<send>(\Delta',m,s)(s_0)\).
  
  Now, we consider the effects of function \(\<send>\) in Byzantine servers and server $b$:
  \begin{enumerate}
    \setcounter{enumi}{1}
    \item \(\bigcup_{s_b \in\byzantineset}
      multiset(\<send>(\Delta,m,s)(s_b).\<sent>) = \bigcup_{s_b
        \in\byzantineset} multiset(\Delta(s_b).\<sent>) =\) \\ \(
      multiset(\Delta'(b).\<sent>) = \<send>(\Delta',m,s)(b).\<sent>\) \\
      since field $\<sent>$ is only modified for server $s$.
    \item Let $s_b$ be a Byzantine server. Given that $m$ is added to the
      $\<pending>$ multiset of all servers it follows that:
      \(\<send>(\Delta',m,s)(b).\<pending> = \Delta'(b).\<pending>
      \cup \{ m \}\) 
    \(\subseteq \Delta(s_b).\<pending> \cup \{m\} =
      \<send>(\Delta,m,s)(s_b).\<pending>\) .
     \item $\<send>(\Delta',m,s)(b).\<received> =
       multiset(\Delta'(b).\<received>)$ \\
       \(\subseteq \bigcup_{s_b\in\byzantineset} multiset(\Delta(s_b).\<received>) =\bigcup_{s_b
         \in\byzantineset} \<send>(\Delta,m,s)(s_b).\<received>\)
       \\ since field $\<received>$ is not modified for any server by function \(\<send>\).
     \item Let $s_b$ be a Byzantine server. Since $m$ is added to the \<pending> multiset of all servers: \\
       \( multiset(\<send>(\Delta,m,s)(s_b).\<received>)
    \cup \<send>(\Delta,m,s)(s_b).\<pending>\)\\ \(= multiset(\Delta(s_b).\<received>)
    \cup \Delta(s_b).\<pending> \cup \{ m \}\)\\\( =
    multiset(\Delta'(b).\<received>) 
    \cup  \Delta'(b).\<pending> \cup \{ m \}\)\\\( =
    multiset(\<send>(\Delta',m,s)(b).\<received>) 
    \cup  \<send>(\Delta',m,s)(b).\<pending> 
    \). 
 
   \item Let $s_b$ be a Byzantine server. Field $\<received>$ is not
     modified for any server, then:\\
       \( multiset(\<send>(\Delta,m,s)(s_b).\<received>)
     = multiset(\Delta(s_b).\<received>)
    \)\\\( \subseteq
    multiset(\Delta'(b).\<received>) 
    =
    multiset(\<send>(\Delta',m,s)(b).\<received>) 
    \).
  \end{enumerate}  
  Consequently, \(\<send>(\Delta,m,s) ~ \<send>(\Delta',m,s)\).

  Next, we will prove that if a correct server receives the same message in two
  observational equivalent networks then the networks remain observational
  equivalents.
  Notice that it is enough to show that (a) \(\<receive>(\Delta,
  m, s)(s).\<received> = \<receive>(\Delta', m, s)(s).\<received>\)
  and \\ (b) \(\<receive>(\Delta, m, s)(s).\<pending> =
  \<receive>(\Delta', m, s)(s).\<pending>\), as these are the only
  fields that the function \(\<receive>\) modifies and therefore the
  remaining conditions follow directly from the relation $\Delta ~
  \Delta'$:  
  \begin{enumerate}[(a)]
    \item \(\<receive>(\Delta, m, s)(s).\<received> =
      \Delta(s).\<received> \cup \{m\} \)\\ \(=  \Delta'(s).\<received> \cup
      \{m\} = \<receive>(\Delta', m, s)(s).\<received>\)
    \item \(\<receive>(\Delta, m, s)(s).\<pending> =
      \Delta(s).\<pending> \setminus \{m\} \) \\ \( =  \Delta'(s).\<pending> \setminus
      \{m\} = \<receive>(\Delta', m, s)(s).\<received>\)
 \end{enumerate}
 In both cases, the first and last equality follow from the definition
 of the function \(\<receive>\) while the second equality from the
 observational equivalence \(\Delta ~ \Delta'\).
 
  Consequently, \(\<receive>(\Delta,m,s) ~ \<receive>(\Delta',m,s)\).
  
\end{proof}  

\begin{lemma}\label{l:net:obs:eq:byz}
  Let \(\Delta\) be a network of model \(\Gamma\), \(\Delta'\) be a
  network of model \(\Gamma'\) and $s$ a Byzantine server. Then,
  \begin{itemize}
    \item \(\<send>(\Delta,m,s) ~ (\<send>(\Delta',m,b) \),
    \item if \(m \in \Delta(s).\<pending>\) and \(m \in
      \Delta'(b).\<pending>\) then \(\<receive>(\Delta,m,s) ~
      \<receive>(\Delta',m,b) \).
    \item if \(m \in \Delta(s).\<pending>\) and \(\Delta'(b).\<received>.\<count>(m) > \Delta(s).\<received>.\<count>(m)\) then\\ \(\<receive>(\Delta,m,s) ~
      \Delta' \).  
  \end{itemize}
\end{lemma}

\begin{proof}
  The proof that if a Byzantine server send a message $m$ in
  networks that is observational equivalent to a network where server
  $b$ send the same messages $m$ then the resulting networks are also observational
  equivalents is analogous to the proof of the previous
  Lemma~\ref{l:net:obs:eq:correct} (where $m$ is sent by the same
  correct server in both networks) except for case (2) where $m$ is
  added to the \(\<sent>\) field of both $b$ and $s$ preserving the
  relation between them:

   \(\bigcup_{s_b \in\byzantineset}
      multiset(\<send>(\Delta,m,s)(s_b).\<sent>) = \bigcup_{s_b
        \in\byzantineset} multiset(\Delta(s_b).\<sent>) \cup \{m\} = \)\\\(
      multiset(\Delta'(b).\<sent>) \cup \{m\}=
      \<send>(\Delta',m,s)(b).\<sent>\)

  Consequently, \(\<send>(\Delta,m,s) ~ (\<send>(\Delta',m,b) \) if
  $s$ is Byzantine.

  Next, we will prove that if a Byzantine server $s_b$ receives a message $m$ in 
  a network that is observational equivalent to a network where server $b$
  receives the same message $m$ then the networks remain observational
  equivalents, as long as, both server have $m$ in their
  \(\<pending>\) multiset:
    \begin{enumerate}
    \item Let $s_0$ be a correct server: \(\<receive>(\Delta, m,
      s_0)(s) = \Delta(s_0) = \Delta'(s_0) = \<receive>(\Delta', m,
      b)(s_0)\) as the server that receives message $m$ is not a
      correct server and therefore for $s_0$ the network is not
      modified.
    \item $\bigcup_{s_b \in\byzantineset}
      multiset(\<receive>(\Delta,m,s)(s_b).\<sent>) = \bigcup_{s_b
        \in\byzantineset} multiset(\Delta(s_b).\<sent>) \)\\\( =
      multiset(\Delta'(b).\<sent>) = multiset(\<receive>(\Delta',m,b)(b).\<sent>)$
      since function \(\<receive>\) does not modify field $\<sent>$
      for any server.
    \item Let $s_b$ be a Byzantine server:
      \begin{itemize}
        \item if $s_b = s$: \(\<receive>(\Delta',m,b)(b).\<pending> = \Delta'(b).\<pending>
      \setminus \{ m \}
    \subseteq \Delta(s).\<pending> \setminus \{m\} =
      \<receive>(\Delta,m,s)(s).\<pending>\) since in both cases $m$ is removed from
      multiset \(\<pending>\) that contained it and therefore the
      inclusion relation between the multisets is preserved.
        \item if $s_b \neq s$: \(\<receive>(\Delta',m,b)(b).\<pending> = \Delta'(b).\<pending>
      \setminus \{ m \}
    \subset \Delta'(b).\<pending>  \subseteq \Delta(s_b).\<pending> =
      \<receive>(\Delta,m,s)(s_b).\<pending>\).
      \end{itemize}
     \item $\<receive>(\Delta',m,b)(b).\<received> =
       multiset(\Delta'(b).\<received>) \cup \{m\}\)\\\( \subseteq
       \bigcup_{s_b\in\byzantineset} multiset(\Delta(s_b).\<received>)
       \cup \{m\} =\bigcup_{s_b
    \in\byzantineset} \<receive>(\Delta,m,s)(s_b).\<received>$\\ since
       message $m$ is added to multiset \(\<received>\) of both $b$
       and $s$.
     \item Let $s_b$ be a Byzantine server.
       For $s_b\neq s$ function
       \(\<receive>(\Delta,m,s)\) does not modify fields
       \(\<receive>\) nor \(\<pending>\) and for $s_b = s$ it moves
       $m$ from its \(\<pending>\) multiset to its \(\<received>\)
       sequence (by hypothesis $m\in\Delta(s).\<pending>$).
       Therefore,  \\
       \( multiset(\<receive>(\Delta,m,s)(s_b).\<received>)
    \cup \<receive>(\Delta,m,s)(s_b).\<pending>\)\\\( = multiset(\Delta(s_b).\<received>)
    \cup \Delta(s_b).\<pending> \).
    Similarly, for $b$ function \(\<receive>(\Delta',m,b)\) moves
       $m$ from its \(\<pending>\) multiset to its \(\<received>\)
       sequence (by hypothesis \\
       $m\in\Delta'(b).\<pending>$).
       Therefore,\\
    \(
    multiset(\Delta'(b).\<received>) 
    \cup  \Delta'(b).\<pending>\)\\\( =
    multiset(\<receive>(\Delta',m,s)(b).\<received>) 
    \cup  \<receive>(\Delta',m,s)(b).\<pending>
    \).
    It follows, that\\ \( multiset(\<receive>(\Delta,m,s)(s_b).\<received>)
    \cup \<receive>(\Delta,m,s)(s_b).\<pending>\)\\\( =
    multiset(\<receive>(\Delta',m,s)(b).\<received>) 
    \cup  \<receive>(\Delta',m,s)(b).\<pending>
    \).
   \item Let $s_b$ be a Byzantine server. Message $m$ is added to the
     \(\<received>\) sequence of the server where the function
     \(\<receive>\) happens and for the remaining servers sequence
     \(\<received>\) is not changed: 
       \( multiset(\<receive>(\Delta,m,s)(s_b).\<received>)
     \subseteq multiset(\Delta(s_b).\<received>) \cup \{m\} \)\\\(
     \subseteq
    multiset(\Delta'(b).\<received>) \cup \{m\}
    =
    multiset(\<receive>(\Delta',m,s)(b).\<received>) 
    \).
  \end{enumerate}

    Next, we will prove that if a Byzantine server $s_b$ receives a message $m$ in 
  a network that is observational equivalent to a network where server $b$
  already received the same message $m$ then the networks remain observational
  equivalents:
    
    \begin{enumerate}
    \item Let $s_0$ be a correct server: \(\<receive>(\Delta, m,
      s_0)(s) = \Delta(s_0) = \Delta'(s_0) \) as the server that receives message $m$ is not a
      correct server and therefore for $s_0$ the network is not
      modified.
    \item  Function \(\<receive>\) does not modify field $\<sent>$
      for any server, then \[\bigcup_{s_b \in\byzantineset}
      multiset(\<receive>(\Delta,m,s)(s_b).\<sent>) = \bigcup_{s_b
        \in\byzantineset} multiset(\Delta(s_b).\<sent>) =
      multiset(\Delta'(b).\<sent>)\].
    \item Let $s_b$ be a Byzantine server:
      \begin{itemize}
        \item if $s_b = s$ then by combining the fact that $b$ and $s$
          were sent the same set of messages (\(multiset(\Delta(s_b).\<received>)
    \cup \Delta(s_b).\<pending>  =
    multiset(\Delta'(b).\<received>)  \cup  \Delta'(b).\<pending> \)) with hypothesis
          \(\Delta'(b).\<received>.\<count>(m) >
          \Delta(s).\<received>.\<count>(m)\)  
     it follows that \(\Delta'(b).\<pending>.\<count>(m) <
      \Delta(s).\<pending>.\<count>(m)\). Hence, removing $m$ from
      \(\Delta(s).\<pending>\) does not break the relation between $b$
      and $s$ \(\<pending>\) sets:
      \(\Delta'(b).\<pending>
    \subseteq \Delta(s).\<pending> \setminus \{m\} =
      \<receive>(\Delta,m,s)(s).\<pending>\).
        \item if $s_b \neq s$ then \(\Delta'(b).\<pending> \subseteq \Delta(s_b).\<pending> =
      \<receive>(\Delta,m,s)(s_b).\<pending>\)
      \end{itemize}
     \item Message $m$ is added to multiset \(\<received>\) of $s$ then:
       $multiset(\Delta'(b).\<received>) \subseteq \)\\\(
       \bigcup_{s_b\in\byzantineset} multiset(\Delta(s_b).\<received>)
       \subset  \bigcup_{s_b\in\byzantineset} multiset(\Delta(s_b).\<received>)
       \cup \{m\} =\)\\\(\bigcup_{s_b
    \in\byzantineset} \<receive>(\Delta,m,s)(s_b).\<received>$.
     \item Let $s_b$ be a Byzantine server. With a reasoning analogous
       to the proof of (5) in the previous case, it can be proven
       that: \\
       \( multiset(\<receive>(\Delta,m,s)(s_b).\<received>)
    \cup \<receive>(\Delta,m,s)(s_b).\<pending> \)\\\(= multiset(\Delta(s_b).\<received>)
    \cup \Delta(s_b).\<pending>  =
    multiset(\Delta'(b).\<received>) 
    \cup  \Delta'(b).\<pending> 
    \). 
 
  \item Let $s_b$ be a Byzantine server.
    \begin{compactitem}
      \item if $s_b = s$: By hypothesis
        $\Delta(s).\<received>.\<count>(m) <
        \Delta'(b).\<received>.\<count>(m)$ therefore adding message $m$
        to the \(\<received>\) sequence of $s$ maintains the
        contaiment relation between the $s$ and $b$ \(\<received>\)
        sequences:
       \( multiset(\<receive>(\Delta,m,s)(s_b).\<received>)
     \subseteq multiset(\Delta(s_b).\<received>) \cup \{m\} \subseteq
    multiset(\Delta'(b).\<received>) 
    \). 
      \item if $s_b \neq s$ then sequence \(\<received>\) is not modified for
        server $s_b$, therefore:\\ \( multiset(\<receive>(\Delta,m,s)(s_b).\<received>)
     =  multiset(\Delta(s_b).\<received>) \) \\ \( 
     \subseteq
    multiset(\Delta'(b).\<received>) 
    \).   
    \end{compactitem}
  \end{enumerate}

\end{proof}

If an event \(\EV\) happens in a correct server $s$ then both models run
the same algorithm.
In particular, the effect of event \(\EV\) in a configuration only
modifies $s$ local state and the network by applying a combination of
the functions \(\<send>\) and \(\<receive>\).
Therefore, if the effect of event \(\EV\) is applied to two
configurations that are observational equivalent, then the local state of
correct state in the resulting configurations remain equal,
Lemma~\ref{l:net:obs:eq:correct} implies that the resulting networks
are observational equivalents, and the history of consensus reached
and the Byzantine knowledge are not changed and hence they remain
equal in the resulting configurations.
In consequence, events that happen in correct servers preserve the observational equivalence
between configurations, as explained in the following Lemma:

\begin{lemma}  \label{l:conf:eq:correct}
  Let \(\Phi = (\Sigma, \Delta,H,K)\) and \(\Psi\) be configurations
  in model \(\Gamma\), 
  \(\Phi' = (\Sigma',\Delta',H',\Tau')\) and \(\Psi'\) be
  configurations in model \(\Gamma'\) and 
  \(\EV\) an event with \(s = \EV.\<server>\) a correct server.
  If \(\Phi ~ \Phi'\), \(\Phi \xrightarrow{\EV} \Psi\) and \(\Phi'
  \xrightarrow{\EV} \Psi'\) then \(\Psi ~ \Psi'\).
\end{lemma}

\begin{proof}
  Since \(\EV\) happens in a correct server its effect only modifies
  the local state of correct servers and the network components of a
  configuration.
  Therefore, \(\Psi = (\Sigma_{\EV},\Delta_{\EV},H,K)\) with
  \[
  \Sigma_{\EV} = \begin{dcases}
    \Sigma \oplus \{s \mapsto (\Sigma(s).S \cup \{x : x =
    e \wedge \<valid>(x) \}, \Sigma(s).H, \Sigma(s).h)\} & \text{if }
    \EV = \<BRB>.\<Deliver>(\<add>(e))\\
   \Sigma \oplus \{s \mapsto (\Sigma(s).S \cup E,
    \Sigma(s).H \cup \{\langle h,E \rangle \}, h)\} & \text{if } \EV =
                   \\ \text{ with }
    E = \{e : e \in \<propset>, \<valid>(e) \wedge e \notin
                   \Sigma(s).H\} &\<SBC>[h].\<SetDeliver>(\<propset>) \\
  \Sigma & \text{otherwise}
                            \end{dcases}
\]

  and \(\Delta_{\EV}\) is generated from \(\Delta\) by applying a
  combination of functions \(\<send>\) and \(\<receive>\) that only depends on
  $\EV$.
  
  Similarly, \(\Psi' = (\Sigma'_{\EV},\Delta'_{\EV},H',\Tau')\) with
  \[
  \Sigma'_{\EV} = \begin{dcases}
    \Sigma' \oplus \{s \mapsto (\Sigma'(s).S \cup \{x : x =
    e \wedge \<valid>(x) \}, & \text{if }
    \EV = \<BRB>.\<Deliver>(\<add>(e)) \\ \Sigma'(s).H, \Sigma'(s).h)\} \\
   \Sigma' \oplus \{s \mapsto (\Sigma'(s).S \cup E,
    \Sigma'(s).H \cup \{\langle h,E \rangle \}, h)\} & \text{if } \EV
                                                       = \<SBC>[h].\<SetDeliver>(\<propset>)
    \\ \text{with } 
    E = \{e : e \in \<propset>, \<valid>(e) \wedge e \notin
    \Sigma'(s).H\}  \\
  \Sigma' & \text{otherwise}
                            \end{dcases}
\]
  and \(\Delta'_{\EV}\) is generated from \(\Delta'\) by applying a
  combination of functions \(\<send>\) and \(\<receive>\) that only depends on
  $\EV$.
  
  To prove that \(\Psi ~ \Psi'\) we need to prove: (1) \(\Sigma_{\EV}
  = \Sigma'_{\EV}\), (2) \(\Delta_{\EV} ~ \Delta'_{\EV}\), (3) \(H =
  H'\) and (4) \(K = \Tau'\):
  The first equality (1) follows from the definition of \(\Sigma_{\EV}\)
  and \(\Sigma'_{\EV}\), and the observational equivalence between
  \(\Phi\) and \(\Phi'\), which implies that \(\Sigma = \Sigma'\).
  For the observational equality between the networks (2) notice that \(\Delta_{\EV}\) is
  generated by applying the same combination of functions \(\<send>\)
  and \(\<receive>\) as the one applied to generate \(\Delta'_{\EV}\) from
  \(\Delta'\).
  Then, since \(\Delta ~ \Delta'\), Lemma ~\ref{l:net:obs:eq:correct}
  can be applied and as consequence \(\Delta_{\EV} ~ \Delta'_{\EV}\). 
  Finally,the observational equivalence between \(\Phi\) and
  \(\Phi'\) directly implies equivalences (3) and (4). 
\end{proof}

Similarly, 
consider an event \(\EV\) that happens in a Byzantine server $s_b$ on
a configuration in model \(\Gamma\) and the same event $\EV'$ except
that it happens in server $b$ on a configuration in model \(\Gamma'\).
The effect of both events in the corresponding configuration is to
just modify the Byzantine knowledge to add valid elements discovered
in the event and apply a combination of the functions \(\<send>\) and
\(\<receive>\) to the network.
Therefore, if both configurations are observational equivalent, then
the local state of correct state and the history of consensus reached
are not changed and hence they remain
equal in the resulting configurations, Lemma~\ref{l:net:obs:eq:byz}
implies that the resulting networks are observational equivalents,
while the same set of valid elements is added to the Byzantine
knowledge sets, preserving the equality between them.
In consequence, event \(\EV\) preserve the observational equivalence
between configurations, as explained in the following Lemma:

\begin{lemma}  \label{l:conf:eq:byz:step}
  Let \(\Phi\) and \(\Psi\) be configurations in model \(\Gamma\),
  \(\Phi'\) and \(\Psi'\) be configurations in model \(\Gamma'\),
  \(\EV\) an event with \(\EV.\<server>\) a Byzantine server and
  \(\EV'\) the same event as \(\EV\) except that \(\EV.\<server> = b\).
  If \(\Phi ~ \Phi'\), \(\Phi \xrightarrow{\EV} \Psi\) and \(\Phi'
  \xrightarrow{\EV'} \Psi'\) then \(\Psi ~ \Psi'\).
\end{lemma}
\begin{proof}
  The proof is analogous to the one of Lemma~\ref{l:conf:eq:correct}.
  The only difference is that the effect of \(\EV\) and \(\EV'\) do
  not modify correct servers local state but rather they modify
  Byzantine knowledge \(K =
  \Tau\) (respectively) by adding \(\ValidElements(\EV)\) to them and
  therefore preserving the equality between them.
\end{proof}

In the following Lemma we show that if there are two reachable
configurations, \(\Phi\) and \(\Phi'\), that are observational
equivalents and an event that happens in Byzantine server $s$ and
represents the reception of a message 
such that server $b$ already received it in \(\Phi'\) but $s$ did not
received it in \(\Phi\) then the configuration that is obtained by
applying the effects of \(\EV\) to \(\Phi\) is observational
equivalent to \(\Phi'\). 

\begin{lemma} \label{l:conf:eq:byz:nostep}
  Let \(\Phi = (\Sigma, \Delta,H,K)\) and \(\Psi\) be
  configurations in model \(\Gamma\), and
  \(\Phi' = (\Sigma',\Delta',H',\Tau')\)  be a reachable configuration
  in model \(\Gamma'\), such that  \(\Phi ~ \Phi'\).
  Consider an event $\EV$ with \(s = \EV.\<server>\) a Byzantine
  server and  \(\Phi \xrightarrow{\EV} \Psi\).
  If \(\EV = \<BRB>.\<Deliver>(m)\) or  \(\EV =
  \<SBC>[h].\<Inform>(\<prop>)\)  but
  \\\(\Delta'(b).\<received>.\<count>(m)>\Delta(s).\<received>.\<count>(m)\)
  then \(\Psi ~ \Phi'\). 

\end{lemma} 
\begin{proof}
  Since \(\EV\) happens in a Byzantine server its effect only modifies
  the network and the Byzantine knowledge components of a
  configuration.
  In particular, event \(\EV\) represents the reception of message
  $m$, therefore \(\Psi = (\Sigma, \<receive>(\Delta,m,s), H, K \cup
  \ValidElements(\EV))\). 
  Hence, to prove that \(\Psi ~ \Phi'\) we need to prove: (1) \(\Sigma
  = \Sigma'\), (2) \(\<receive>(\Delta,m,s) ~ \Delta'\), (3) \(H =
  H'\) and (4) \(K \cup \ValidElements(\EV) = \Tau\).
  The observational equivalence between \(\Phi\) and
  \(\Phi'\) directly implies equivalences (1) and (3).
  Lemma~\ref{l:net:obs:eq:byz} implies that \(\<receive>(\Delta,m,s) ~
  \Delta'\) as \(\Delta ~ \Delta'\) follows from the observational
  equivalence \(\Phi ~ \Phi'\) and by hypothesis
  \(\Delta'(b).\<received>.\<count>(m)>\Delta(s).\<received>.\<count>(m)\).  
  Finally, to prove (4) notice that 
  \(\Delta'(b).\<received>.\<count>(m)>\Delta(s).\<received>.\<count>(m)
  \geq 0\) implies that  \(m\in\Delta'(b).\<received>\). 
  Then, by Lemma~\ref{l:receive:tau}, \(\ValidElements(\EV) \subseteq
  \Tau\).
  Therefore, \(K \cup \ValidElements(\EV) = \Tau \cup
  \ValidElements(\EV) = \Tau\).     
  
\end{proof}

From the previous three Lemmas it follows that for each valid trace in
one model there is a valid trace in the other model such that
corresponding configurations are indistinguishable, as any event that
happens in one trace can be mapped to one (or $f$) events in the
other.

\begin{theorem}\label{thm:simpl}
  For every valid trace in model $\Gamma$ there exists a valid
  trace in model $\Gamma'$ such that the corresponding state are
  indistinguishable. 
\end{theorem}

\begin{proof}
  Let $\sigma$ be a valid trace in model $\Gamma$.
  We need to construct a valid trace $\sigma'$ in model $\Gamma'$ such
  that $\sigma_i ~ \sigma'_i$ for all $i \geq 0$.
  
  We define \(\sigma'_0 = \Phi'_0\) as the initial configuration in
  $\Gamma'$, which has all its components empty and therefore is
  observational equivalent to the initial configuration in $\Gamma$,
  \(\Phi_0 = \sigma_0\).
  
  Next, for $n \geq 0$ assume that $\sigma'$ is defined up to $n$ and
  $\sigma_i ~ \sigma'_i$ if $n>i\geq0$.
  Consider the event \(\EV_n\) such that \(\sigma_n
  \xrightarrow{\EV_n} \sigma_{n+1}\), with \(\sigma_n = (\Sigma,
  \Delta, H,K)\).
  We will show that exists an event \(\EV'_n\) and a configuration
  \(\sigma'_{n+1}\) such that \(\sigma'_n
  \xrightarrow{\EV'_n} \sigma'_{n+1}\) and \(\sigma'_{n+1} ~
  \sigma_{n+1}\), with \(\sigma'_n = (\Sigma',\Delta', H',\Tau')\).
  
  Proceed by case analysis on \(\EV_n\):
  \begin{itemize}
   \item if \(s = \EV_n.\<server>\) is a correct server then take \(\EV'_n =
     \EV_n\).
     Therefore, \(\EV_n'\) is enabled in \(\sigma_n'\) because whether an
     event that happens in a correct server is enabled or not in
     \(\sigma_n'\) depends only on $s$ local
     state, $s$ network and the history of consensus reached, 
     and all these components are the same in both configuration
     since \(\sigma_n ~ \sigma_n'\).
     %
     As consequence, exists \(\sigma_{n+1}'\) such that
     \(\sigma_{n}'\xrightarrow{\EV'}\sigma_{n+1}'\) and from
     Lemma~\ref{l:conf:eq:correct} follows that \(\sigma_{n}~\sigma_{n+1}'\).
     
   \item if \(s = \EV_n.\<server>\) is Byzantine server then consider
     \(\EV_n'\) as the same event as \(\EV_n\) but with \(\EV_n.\<server> =
     b\). 
     Two things can happen, either \(\EV_n'\) is enabled in
     \(\sigma_{n}'\) or not.
     In the former case, there exists \(\sigma'_{n+1}\) such that
     \(\sigma_{n}'\xrightarrow{\EV_n'}\sigma'_{n+1}\) and from
     Lemma~\ref{l:conf:eq:byz:step} follows that \(\sigma_{n+1}~\sigma_{n+1}'\), as
     desired.
     In the latter case, it must be the case that \(\EV\) represent
     the reception of message that is pending in \(s\) but not in
     \(b\). In symbols, \(\EV \in \Delta(s).\<pending>\) and \(\EV'
     \notin \Delta'(b).\<pending>\).
     Since \(\Delta ~ \Delta'\), it follows that
     \(multiset(\Delta(s).\<received>).\<count>(\EV) < (multiset(\Delta(s).\<received>) \cup \Delta(s).\<pending>
     ).\<count>(\EV) = (multiset(\Delta'(b).\<received>) \cup \Delta'(b).\<pending>
     ).\<count>(\EV') =\)\\\( multiset(\Delta(s).\<received>).\<count>(\EV)\). 
     Hence, Lemma~\ref{l:conf:eq:byz:nostep} implies that
     \(\sigma_{n+1}~\sigma_{n}'\).
     Then, take \(\EV_{n}'\) as \(\<nop>\), which is enabled in
     \(\sigma_{n}'\), and \(\sigma_{n+1}' = \sigma_{n}'\).  
   \item if \(\EV_{n} = \<SBC>[h].\<Consensus>(\<propset>)\) then \(\sigma_{n+1} =
     (\Sigma,\Delta,H \oplus \{h \mapsto \<propset>\},K)\).
     Notice that \(\EV'_{n}= \EV_{n}\) is enabled in \(\sigma_{n}'\). 
     In fact, since \(\sigma_{n}~\sigma_{n}'\), \(H=H'\) and \(\Delta ~\Delta'\) 
     which implies that \((h,p) \in \Delta(s).\<sent>\) for some
     $s$ if and only if \((h,p) \in \Delta'(s').\<sent>\) for some
     $s'$.
     Therefore, \(\sigma_{n}'\xrightarrow{\EV_{n}'} (\Sigma',\Delta',H'
     \oplus \{h \mapsto propset\},\Tau')\) which is observational
     equivalent to \(\sigma_{n+1}\).
   \item if \(\EV = \<nop>\) then \(\sigma_{n+1} = \sigma_{n}\).
     Therefore, by taking
     \(\EV'_{n}\) as \(\<nop>\), which is enabled in
     \(\sigma'_{n}\), and \(\sigma_{n+1}' = \sigma'_{n}\) follows that
     \(\sigma_{n}'\xrightarrow{\EV'}\sigma_{n+1}'\) and \(\sigma_{n+1} ~ \sigma_{n+1}'\).  
  \end{itemize}
  Hence, the theorem follows by induction.
\end{proof}
We now prove the other direction.
For that, we first define the \emph{stuttering extension} of a trace
\(\sigma\) as the trace \(\sigma^{st}\) which adds \(f-1\) events
\(\<nop>\) after each event in \(\sigma\).
%
\begin{theorem}
  For every valid trace $\sigma'$ in model $\Gamma'$ there
  exists a valid trace $\sigma$ in model $\Gamma$ such that
  each state in $\sigma$ is indistinguishable with corresponding state
  in the stuttering extension of $\sigma'$ .
\end{theorem}

\begin{proof}
  Let $\sigma'$ be a valid trace in model $\Gamma'$.
  We need to construct a valid trace $\sigma$ in model $\Gamma$ such
  that $\sigma_i ~ \sigma'^{st}_i$ for $i\geq0$, or equivalently,
  \(\sigma_{i*f+j} ~ \sigma'_{i}\)
  for $0\leq j<f$.
  
  We start by defining \(\sigma_j = \Phi_0\), for $0\leq j<f$, as the initial
  configuration in $\Gamma$, which has all its components empty and
  therefore is observational equivalent to the initial configuration
  in $\Gamma'$, \(\Phi'_0 = \sigma'_0\).
  
  Next, for $n \geq 0$ assume that $\sigma$ is defined up to $k = n*f+f-1$ and
  that \(\sigma_{i*f+j} ~
  \sigma'_{i}\) for $0 \leq i \leq n$ and $0\leq j<f$.
  
  We will show that exist $f$ events \(\EV_{k}, \ldots, \EV_{k+f-1}\)
  and $f$ configurations \(\sigma_{k+1},\ldots,\sigma_{k+f}\) such
  that \(\sigma_{k+j}
  \xrightarrow{\EV_{k+j}} \sigma_{k+j+1} \) for any $0\leq j<f$.
    
  Proceed by case analysis on the event \(\EV'_n\) such that \(\sigma'_n
  \xrightarrow{\EV'_n} \sigma'_{n+1}\):
  
  \begin{itemize}
   \item if \(\EV'_n.\<server> = b\) with \(\EV'_n =
     \<BRB>.\<Deliver>(m)\) or \(\EV'_n = \<SBC>[h].\<Inform>(prop)\)
     then take \(\EV_{k+j} = 
     \EV'_n\) but with \(\EV_{k+j}.\<server> = b_j\) for $0\leq j<f$.
     Therefore, \(\EV_{k}\) is enabled in
     \(\sigma_{k}\) because $m\in\sigma_n'(b).\<pending>
     \subseteq\sigma_{k}(b_0).\<pending>$.
     %
     As consequence, exists \(\sigma_{k+1}\) such that
     \(\sigma_{k}'\xrightarrow{\EV_{k}}\sigma_{k+1}'\)
     and from Lemma~\ref{l:conf:eq:byz:step} follows that
     \(\sigma_{k+1}~\sigma_{n+1}'\).
     
     Since \(\EV'_n\) modifies the network by moving $m$ from
     the \(\<pending>\) multiset of $b$ to its \(\<received>\) sequence it
     follows that \(\sigma'_{n+1}(b).\<received>.\<count>(m) =
     \sigma'_{n}(b).\<received>.\<count>(m)+1 >
     \sigma'_{n}(b).\<received>.\<count>(m) \geq
     \sigma_{k}(b_j).\<received>.\<count>(m) =
     \sigma_{k+1}(b_j).\<received>.\<count>(m)\) for \(0 < j < f\),
     where the last equality follows from the fact that
     \(\EV_{k}\) only modifies the \(\<pending>\) multiset of
     \(b_0\) and the observational equivalence between
     \(\sigma_{k}\) and \(\sigma'_{n}\).
     By combining this result with the fact that
     \(multiset(\sigma_{k+1}(b_1).\<received>) \cup
     \sigma_{k+1}(b_1).\<pending> = multiset(\sigma'_{n}(b).\<received>) \cup
     \sigma'_{n}(b).\<pending>\) we get that $m \in
     \sigma_{k+1}(b_1).\<pending>$ and therefore \(\EV_{k+1}\) is enabled in \(\sigma_{k+1}\).  
     As consequence, exists \(\sigma_{k+2}\) such that
     \(\sigma_{k+1}\xrightarrow{\EV_{k+1}}\sigma_{k+2}\)
     and from Lemma~\ref{l:conf:eq:byz:nostep} follows that
     \(\sigma_{k+2}~\sigma_{n+1}'\).
     By repeating this reasoning for $1<j< f$ it can be proven that
     exists \(\sigma_{k+j+1}\) such that
     \(\sigma_{k+j}\xrightarrow{\EV_{k+j}}\sigma_{k+j+1}\)
     and \(\sigma_{k+j+1}~\sigma_{n+1}'\).
   \item otherwise, take \(\EV_{k} = \EV'_{n}\).
     Then, with an analysis analogous to the one in the previous theorem,
     it can be shown that \(\EV_{k}\) is enabled in
     \(\sigma_{k}\) and, moreover, the resulting configuration
     \(\sigma_{k} \xrightarrow{\EV_{k}} \sigma_{k+1}\) satisfies
     \(\sigma_{k+1} ~ \sigma'_{n+1}\).
     Finally, for \(1<j<f\) define \(\EV_{k+j} = \<nop>\) 
     which is enabled in all configurations, in particular, in
     \(\sigma_{k+j}\).
     Therefore, \(\sigma_{k+j}
     \xrightarrow{\EV_{k+j}}\sigma_{k+j+1}\) and
     \(\sigma_{k+j+1} ~ \sigma'_{n+1}\), as desired.
  \end{itemize}
  Hence, the theorem follows by induction.
\end{proof}

As consequence, the main result of this Section is that properties
for model $\Gamma$ hold if and only if they also hold in model $\Gamma'$.
That is, one can prove properties for traces in model $\Gamma'$
where all the components are well-defined, and the result will
directly translate to traces in model $\Gamma'$.


\section{Concluding Remarks and Future Work}\label{sec:discussion}

We presented in this paper a novel distributed data-type, called
\setchain, that implements a grow-only set with epochs, and tolerates
Byzantine server nodes.
We provided a low-level specification of desirable properties of
\setchains and three distributed implementations, where the most
efficient one uses Byzantine Reliable Broadcast and RedBelly Set
Byzantine Consensus.
Our empirical evaluation suggests that the performance of inserting
elements in \setchain is three orders of magnitude faster than
consensus.
Also, we have proved that the behavior of the Byzantine server nodes
can be modeled by a collection of simple interactions with BRB and
SBC, and we have introduced a non-deterministic process that
encompasses these interactions.
This modeling paves the way to use formal reasoning that is not
equipped for Byzantine reasoning to reason about \setchain.

Future work includes developing the motivating applications listed in
the introduction, including mempool logs using \setchain\!\!s, and L2
faster optimistic rollups.
\setchain can also be used to alleviate front-running attacks.
The mempool stores the transactions requested by users, so observing
the mempool allows to predict future operations.
\emph{Front-running} is the action of observing transaction request and
maliciously inject transactions to be executed before the observed
ones~\cite{Daian2020FlashBoys,Ferreira2021Frontrunner} (by paying a
higher fee to miners).
\setchain can be used to \emph{detect} front-running since it can
serve as a basic mechanism to build a mempool that is efficient and
serves as a log of requests.
Additionally, \setchains can be used as a building block to solve front-running
where users encrypt their requests using a multi-signature decryption scheme,
where participant decrypting servers decrypt requests after they are chosen for
execution by miners once the order has already been fixed.

We will also study how to equip blockchains with \setchain
(synchronizing blocks and epochs) to allow smart-contracts to access
the \setchain as part of their storage.

An important remaining problem is how to design a payment system for
clients to pay for the usage of \setchain (even if a much smaller fee
than for the blockchain itself).
Our \setchain exploits a specific partial orders that relaxes the
total order imposed by blockchains.
As future work, we will explore other partial orders and their uses,
for example, federations of \epochs, one \epochs per smart-contract,
etc.


%

%
%
%
%


\bibliographystyle{ACM-Reference-Format}
\bibliography{bibfile}


\begin{thebibliography}{35}


\ifx \showCODEN    \undefined \def \showCODEN     #1{\unskip}     \fi
\ifx \showDOI      \undefined \def \showDOI       #1{#1}\fi
\ifx \showISBNx    \undefined \def \showISBNx     #1{\unskip}     \fi
\ifx \showISBNxiii \undefined \def \showISBNxiii  #1{\unskip}     \fi
\ifx \showISSN     \undefined \def \showISSN      #1{\unskip}     \fi
\ifx \showLCCN     \undefined \def \showLCCN      #1{\unskip}     \fi
\ifx \shownote     \undefined \def \shownote      #1{#1}          \fi
\ifx \showarticletitle \undefined \def \showarticletitle #1{#1}   \fi
\ifx \showURL      \undefined \def \showURL       {\relax}        \fi
\providecommand\bibfield[2]{#2}
\providecommand\bibinfo[2]{#2}
\providecommand\natexlab[1]{#1}
\providecommand\showeprint[2][]{arXiv:#2}

\bibitem[Ben-Sasson et~al\mbox{.}(2014a)]%
        {Sasson2014ZeroCash}
\bibfield{author}{\bibinfo{person}{Eli Ben-Sasson}, \bibinfo{person}{Alessandro
  Chiesa}, \bibinfo{person}{Christina Garman}, \bibinfo{person}{Matthew Green},
  \bibinfo{person}{Ian Miers}, \bibinfo{person}{Eran Tromer}, {and}
  \bibinfo{person}{Madars Virza}.} \bibinfo{year}{2014}\natexlab{a}.
\newblock \showarticletitle{Zerocash: Decentralized Anonymous Payments from
  Bitcoin}. In \bibinfo{booktitle}{\emph{Proc. of S\&P'14}}.
  \bibinfo{pages}{459--474}.
\newblock
\urldef\tempurl%
\url{https://doi.org/10.1109/SP.2014.36}
\showDOI{\tempurl}


\bibitem[Ben-Sasson et~al\mbox{.}(2014b)]%
        {Sasson2014ZKvonNeumann}
\bibfield{author}{\bibinfo{person}{Eli Ben-Sasson}, \bibinfo{person}{Alessandro
  Chiesa}, \bibinfo{person}{Eran Tromer}, {and} \bibinfo{person}{Madars
  Virza}.} \bibinfo{year}{2014}\natexlab{b}.
\newblock \showarticletitle{Succinct Non-Interactive {Z}ero {K}nowledge for a
  von {N}eumann Architecture}. In \bibinfo{booktitle}{\emph{Proc. of USENIX
  Sec.'14}}. \bibinfo{publisher}{USENIX}, \bibinfo{pages}{781--796}.
\newblock
\showISBNx{978-1-931971-15-7}
\urldef\tempurl%
\url{https://www.usenix.org/conference/usenixsecurity14/technical-sessions/presentation/ben-sasson}
\showURL{%
\tempurl}


\bibitem[Bracha(1987)]%
        {DBLP:journals/iandc/Bracha87}
\bibfield{author}{\bibinfo{person}{Gabriel Bracha}.}
  \bibinfo{year}{1987}\natexlab{}.
\newblock \showarticletitle{Asynchronous Byzantine Agreement Protocols}.
\newblock \bibinfo{journal}{\emph{Inf. Comput.}} \bibinfo{volume}{75},
  \bibinfo{number}{2} (\bibinfo{year}{1987}), \bibinfo{pages}{130--143}.
\newblock
\urldef\tempurl%
\url{https://doi.org/10.1016/0890-5401(87)90054-X}
\showDOI{\tempurl}


\bibitem[Chandra and Toueg(1996)]%
        {Chandra1998FailureBAB}
\bibfield{author}{\bibinfo{person}{Tushar~Deepak Chandra} {and}
  \bibinfo{person}{Sam Toueg}.} \bibinfo{year}{1996}\natexlab{}.
\newblock \showarticletitle{Unreliable Failure Detectors for Reliable
  Distributed Systems}.
\newblock \bibinfo{journal}{\emph{J. ACM}} \bibinfo{volume}{43},
  \bibinfo{number}{2} (\bibinfo{date}{mar} \bibinfo{year}{1996}),
  \bibinfo{pages}{225–267}.
\newblock
\showISSN{0004-5411}
\urldef\tempurl%
\url{https://doi.org/10.1145/226643.226647}
\showDOI{\tempurl}


\bibitem[Cholvi et~al\mbox{.}(2021)]%
        {Cholvi2021BDSO}
\bibfield{author}{\bibinfo{person}{Vicent Cholvi}, \bibinfo{person}{Antonio
  {Fernández Anta}}, \bibinfo{person}{Chryssis Georgiou},
  \bibinfo{person}{Nicolas Nicolaou}, \bibinfo{person}{Michel Raynal}, {and}
  \bibinfo{person}{Antonio Russo}.} \bibinfo{year}{2021}\natexlab{}.
\newblock \showarticletitle{Byzantine-tolerant Distributed Grow-only Sets:
  Specification and Applications}. In \bibinfo{booktitle}{\emph{Proc. of
  FAB'21}}. \bibinfo{pages}{2:1–2:19}.
\newblock


\bibitem[Crain et~al\mbox{.}(2021)]%
        {Crain2021RedBelly}
\bibfield{author}{\bibinfo{person}{Tyler Crain}, \bibinfo{person}{Christopher
  Natoli}, {and} \bibinfo{person}{Vincent Gramoli}.}
  \bibinfo{year}{2021}\natexlab{}.
\newblock \showarticletitle{Red Belly: A Secure, Fair and Scalable Open
  Blockchain}. In \bibinfo{booktitle}{\emph{Proc. of S\&P'21}}.
  \bibinfo{pages}{466--483}.
\newblock
\urldef\tempurl%
\url{https://doi.org/10.1109/SP40001.2021.00087}
\showDOI{\tempurl}


\bibitem[Cristian et~al\mbox{.}(1995)]%
        {Cristin1996AtomicBroadcast}
\bibfield{author}{\bibinfo{person}{F. Cristian}, \bibinfo{person}{H. Aghili},
  \bibinfo{person}{R. Strong}, {and} \bibinfo{person}{D. Volev}.}
  \bibinfo{year}{1995}\natexlab{}.
\newblock \showarticletitle{Atomic Broadcast: from simple message diffusion to
  Byzantine agreement}. In \bibinfo{booktitle}{\emph{25th Int'l Symp. on
  Fault-Tolerant Computing}}. \bibinfo{pages}{431--}.
\newblock
\urldef\tempurl%
\url{https://doi.org/10.1109/FTCSH.1995.532668}
\showDOI{\tempurl}


\bibitem[Croman et~al\mbox{.}(2016)]%
        {Croman2016ScalingDecentralizedBlockchain}
\bibfield{author}{\bibinfo{person}{Kyle Croman}, \bibinfo{person}{Christian
  Decker}, \bibinfo{person}{Ittay Eyal}, \bibinfo{person}{Adem~Efe Gencer},
  \bibinfo{person}{Ari Juels}, \bibinfo{person}{Ahmed Kosba},
  \bibinfo{person}{Andrew Miller}, \bibinfo{person}{Prateek Saxena},
  \bibinfo{person}{Elaine Shi}, \bibinfo{person}{Emin G{\"u}n~Sirer},
  \bibinfo{person}{Dawn Song}, {and} \bibinfo{person}{Roger Wattenhofer}.}
  \bibinfo{year}{2016}\natexlab{}.
\newblock \showarticletitle{On Scaling Decentralized Blockchains}. In
  \bibinfo{booktitle}{\emph{Financial Crypto. and Data Security}}.
  \bibinfo{publisher}{Springer}, \bibinfo{pages}{106--125}.
\newblock
\showISBNx{978-3-662-53357-4}


\bibitem[Daian et~al\mbox{.}(2020)]%
        {Daian2020FlashBoys}
\bibfield{author}{\bibinfo{person}{Philip Daian}, \bibinfo{person}{Steven
  Goldfeder}, \bibinfo{person}{T. Kell}, \bibinfo{person}{Yunqi Li},
  \bibinfo{person}{X. Zhao}, \bibinfo{person}{Iddo Bentov},
  \bibinfo{person}{Lorenz Breidenbach}, {and} \bibinfo{person}{A. Juels}.}
  \bibinfo{year}{2020}\natexlab{}.
\newblock \showarticletitle{Flash Boys 2.0: Frontrunning in Decentralized
  Exchanges, Miner Extractable Value, and Consensus Instability}.
\newblock \bibinfo{journal}{\emph{Proc. of S\&P'20}} (\bibinfo{year}{2020}),
  \bibinfo{pages}{910--927}.
\newblock


\bibitem[Dang et~al\mbox{.}(2019)]%
        {Dang2019Sharding}
\bibfield{author}{\bibinfo{person}{Hung Dang}, \bibinfo{person}{Tien Tuan~Anh
  Dinh}, \bibinfo{person}{Dumitrel Loghin}, \bibinfo{person}{Ee-Chien Chang},
  \bibinfo{person}{Qian Lin}, {and} \bibinfo{person}{Beng~Chin Ooi}.}
  \bibinfo{year}{2019}\natexlab{}.
\newblock \showarticletitle{Towards Scaling Blockchain Systems via Sharding}.
  In \bibinfo{booktitle}{\emph{Proc. of SIGMOD'19}}. \bibinfo{publisher}{ACM},
  \bibinfo{pages}{123–--140}.
\newblock
\showISBNx{9781450356435}
\urldef\tempurl%
\url{https://doi.org/10.1145/3299869.3319889}
\showDOI{\tempurl}


\bibitem[D\'{e}fago et~al\mbox{.}(2004)]%
        {Defago2004BAB}
\bibfield{author}{\bibinfo{person}{Xavier D\'{e}fago},
  \bibinfo{person}{Andr\'{e} Schiper}, {and} \bibinfo{person}{P\'{e}ter
  Urb\'{a}n}.} \bibinfo{year}{2004}\natexlab{}.
\newblock \showarticletitle{Total Order Broadcast and Multicast Algorithms:
  Taxonomy and Survey}.
\newblock \bibinfo{journal}{\emph{ACM Comput. Surv.}} \bibinfo{volume}{36},
  \bibinfo{number}{4} (\bibinfo{date}{dec} \bibinfo{year}{2004}),
  \bibinfo{pages}{372–421}.
\newblock
\showISSN{0360-0300}
\urldef\tempurl%
\url{https://doi.org/10.1145/1041680.1041682}
\showDOI{\tempurl}


\bibitem[Donovan and Kernighan(2015)]%
        {donovan15go}
\bibfield{author}{\bibinfo{person}{Alan~A.A. Donovan} {and}
  \bibinfo{person}{Brian~W. Kernighan}.} \bibinfo{year}{2015}\natexlab{}.
\newblock \bibinfo{booktitle}{\emph{The {Go} Programming Language}}.
\newblock \bibinfo{publisher}{Adison-Wesley}.
\newblock


\bibitem[{Fern{\'a}ndez Anta} et~al\mbox{.}(2021)]%
        {anta2021principles}
\bibfield{author}{\bibinfo{person}{Antonio {Fern{\'a}ndez Anta}},
  \bibinfo{person}{Chryssis Georgiou}, \bibinfo{person}{Maurice Herlihy}, {and}
  \bibinfo{person}{Maria Potop-Butucaru}.} \bibinfo{year}{2021}\natexlab{}.
\newblock \bibinfo{booktitle}{\emph{Principles of Blockchain Systems}}.
\newblock \bibinfo{publisher}{Morgan \& Claypool Publishers}.
\newblock


\bibitem[{Fern{\'a}ndez Anta} et~al\mbox{.}(2018)]%
        {anta2018formalizing}
\bibfield{author}{\bibinfo{person}{Antonio {Fern{\'a}ndez Anta}},
  \bibinfo{person}{Kishori Konwar}, \bibinfo{person}{Chryssis Georgiou}, {and}
  \bibinfo{person}{Nicolas Nicolaou}.} \bibinfo{year}{2018}\natexlab{}.
\newblock \showarticletitle{Formalizing and implementing distributed ledger
  objects}.
\newblock \bibinfo{journal}{\emph{ACM Sigact News}} \bibinfo{volume}{49},
  \bibinfo{number}{2} (\bibinfo{year}{2018}), \bibinfo{pages}{58--76}.
\newblock


\bibitem[Fischer et~al\mbox{.}(1985)]%
        {Fischer1985Impossibility}
\bibfield{author}{\bibinfo{person}{Michael~J. Fischer},
  \bibinfo{person}{Nancy~A. Lynch}, {and} \bibinfo{person}{Michael~S.
  Paterson}.} \bibinfo{year}{1985}\natexlab{}.
\newblock \showarticletitle{Impossibility of Distributed Consensus with One
  Faulty Process}.
\newblock \bibinfo{journal}{\emph{JACM}} \bibinfo{volume}{32},
  \bibinfo{number}{2} (\bibinfo{year}{1985}), \bibinfo{pages}{374--382}.
\newblock
\showISSN{0004-5411}
\urldef\tempurl%
\url{https://doi.org/10.1145/3149.214121}
\showDOI{\tempurl}


\bibitem[Guerraoui et~al\mbox{.}(2019)]%
        {DBLP:conf/podc/GuerraouiKMPS19}
\bibfield{author}{\bibinfo{person}{Rachid Guerraoui}, \bibinfo{person}{Petr
  Kuznetsov}, \bibinfo{person}{Matteo Monti}, \bibinfo{person}{Matej Pavlovic},
  {and} \bibinfo{person}{Dragos{-}Adrian Seredinschi}.}
  \bibinfo{year}{2019}\natexlab{}.
\newblock \showarticletitle{The Consensus Number of a Cryptocurrency}. In
  \bibinfo{booktitle}{\emph{Proc. of PODC'19}}. \bibinfo{publisher}{{ACM}},
  \bibinfo{pages}{307--316}.
\newblock
\urldef\tempurl%
\url{https://doi.org/10.1145/3293611.3331589}
\showDOI{\tempurl}


\bibitem[Jourenko et~al\mbox{.}(2019)]%
        {Jourenko2019SoKAT}
\bibfield{author}{\bibinfo{person}{Maxim Jourenko}, \bibinfo{person}{Kanta
  Kurazumi}, \bibinfo{person}{Mario Larangeira}, {and} \bibinfo{person}{Keisuke
  Tanaka}.} \bibinfo{year}{2019}\natexlab{}.
\newblock \showarticletitle{SoK: A Taxonomy for Layer-2 Scalability Related
  Protocols for Cryptocurrencies}.
\newblock \bibinfo{journal}{\emph{IACR Cryptol. ePrint Arch.}}
  \bibinfo{volume}{2019} (\bibinfo{year}{2019}), \bibinfo{pages}{352}.
\newblock


\bibitem[Kalodner et~al\mbox{.}(2018)]%
        {Kalodner2018Arbitrum}
\bibfield{author}{\bibinfo{person}{Harry Kalodner}, \bibinfo{person}{Steven
  Goldfeder}, \bibinfo{person}{Xiaoqi Chen}, \bibinfo{person}{S.~Matthew
  Weinberg}, {and} \bibinfo{person}{Edward~W. Felten}.}
  \bibinfo{year}{2018}\natexlab{}.
\newblock \showarticletitle{Arbitrum: Scalable, private smart contracts}. In
  \bibinfo{booktitle}{\emph{27th {USENIX} Security Symposium}}.
  \bibinfo{publisher}{{USENIX} Assoc.}, \bibinfo{pages}{1353--1370}.
\newblock
\showISBNx{978-1-939133-04-5}
\urldef\tempurl%
\url{https://www.usenix.org/conference/usenixsecurity18/presentation/kalodner}
\showURL{%
\tempurl}


\bibitem[Kwon and Buchman(2019)]%
        {kwon2019cosmos}
\bibfield{author}{\bibinfo{person}{Jae Kwon} {and} \bibinfo{person}{Ethan
  Buchman}.} \bibinfo{year}{2019}\natexlab{}.
\newblock \bibinfo{title}{Cosmos whitepaper}.
\newblock
\newblock


\bibitem[Mahdi et~al\mbox{.}(2018)]%
        {Zamani2018RapidChain}
\bibfield{author}{\bibinfo{person}{Zamani Mahdi}, \bibinfo{person}{Mahnush
  Movahedi}, {and} \bibinfo{person}{Mariana Raykova}.}
  \bibinfo{year}{2018}\natexlab{}.
\newblock \showarticletitle{RapidChain: Scaling Blockchain via Full Sharding}.
  In \bibinfo{booktitle}{\emph{Proc. of CSS'18}} (Toronto, Canada).
  \bibinfo{publisher}{ACM}, \bibinfo{pages}{931–--948}.
\newblock
\showISBNx{9781450356930}
\urldef\tempurl%
\url{https://doi.org/10.1145/3243734.3243853}
\showDOI{\tempurl}


\bibitem[Nakamoto(2009)]%
        {nakamoto06bitcoin}
\bibfield{author}{\bibinfo{person}{Satoshi Nakamoto}.}
  \bibinfo{year}{2009}\natexlab{}.
\newblock \bibinfo{title}{{B}itcoin: a peer-to-peer electronic cash system}.
\newblock
\newblock


\bibitem[Poon and Dryja(2016)]%
        {Poon2016lightning}
\bibfield{author}{\bibinfo{person}{Joseph Poon} {and} \bibinfo{person}{Thaddeus
  Dryja}.} \bibinfo{year}{2016}\natexlab{}.
\newblock \bibinfo{booktitle}{\emph{The Bitcoin Lightning Network: Scalable
  Off-Chain Instant Payments}}.
\newblock
\urldef\tempurl%
\url{https://lightning.network/lightning-network-paper.pdf}
\showURL{%
\tempurl}


\bibitem[Raynal(2018)]%
        {Raynal2018FaultTolerantMessagePassing}
\bibfield{author}{\bibinfo{person}{Michel Raynal}.}
  \bibinfo{year}{2018}\natexlab{}.
\newblock \bibinfo{booktitle}{\emph{Fault-Tolerant Message-Passing Distributed
  Systems: An Algorithmic Approach}}.
\newblock
\showISBNx{978-3-319-94140-0}
\urldef\tempurl%
\url{https://doi.org/10.1007/978-3-319-94141-7}
\showDOI{\tempurl}


\bibitem[{Robinson, Dan} and {Konstantopoulos, Georgios}(2020)]%
        {Robinson2020DarkForest}
\bibfield{author}{\bibinfo{person}{{Robinson, Dan}} {and}
  \bibinfo{person}{{Konstantopoulos, Georgios}}.}
  \bibinfo{year}{2020}\natexlab{}.
\newblock \bibinfo{booktitle}{\emph{Ethereum is a Dark Forest}}.
\newblock
\urldef\tempurl%
\url{https://medium.com/@danrobinson/ethereum-is-a-dark-forest-ecc5f0505dff}
\showURL{%
\tempurl}


\bibitem[Saad et~al\mbox{.}(2019)]%
        {Saad2019DDoSMempool}
\bibfield{author}{\bibinfo{person}{Muhammad Saad}, \bibinfo{person}{Laurent
  Njilla}, \bibinfo{person}{Charles Kamhoua}, \bibinfo{person}{Joongheon Kim},
  \bibinfo{person}{DaeHun Nyang}, {and} \bibinfo{person}{Aziz Mohaisen}.}
  \bibinfo{year}{2019}\natexlab{}.
\newblock \showarticletitle{Mempool optimization for Defending Against {DDoS}
  Attacks in {PoW}-based Blockchain Systems}. In
  \bibinfo{booktitle}{\emph{Proc. of ICBC'19}}. \bibinfo{pages}{285--292}.
\newblock
\urldef\tempurl%
\url{https://doi.org/10.1109/BLOC.2019.8751476}
\showDOI{\tempurl}


\bibitem[Saad et~al\mbox{.}(2018)]%
        {Saad2018DDoSMempool}
\bibfield{author}{\bibinfo{person}{Muhammad Saad}, \bibinfo{person}{My~T.
  Thai}, {and} \bibinfo{person}{Aziz Mohaisen}.}
  \bibinfo{year}{2018}\natexlab{}.
\newblock \showarticletitle{{POSTER}: Deterring {DDoS} Attacks on
  Blockchain-Based Cryptocurrencies through Mempool Optimization}. In
  \bibinfo{booktitle}{\emph{Proc. of ASIACCS'18}}. \bibinfo{publisher}{ACM},
  \bibinfo{pages}{809--–811}.
\newblock
\showISBNx{9781450355766}
\urldef\tempurl%
\url{https://doi.org/10.1145/3196494.3201584}
\showDOI{\tempurl}


\bibitem[Shapiro et~al\mbox{.}(2011)]%
        {Shapiro2011CCRDT}
\bibfield{author}{\bibinfo{person}{Marc Shapiro}, \bibinfo{person}{Nuno
  Pregui{\c c}a}, \bibinfo{person}{Carlos Baquero}, {and}
  \bibinfo{person}{Marek Zawirski}.} \bibinfo{year}{2011}\natexlab{}.
\newblock \showarticletitle{{Convergent and Commutative Replicated Data
  Types}}.
\newblock \bibinfo{journal}{\emph{{Bulletin- European Association for
  Theoretical Computer Science}}} \bibinfo{number}{104} (\bibinfo{date}{June}
  \bibinfo{year}{2011}), \bibinfo{pages}{67--88}.
\newblock
\urldef\tempurl%
\url{https://hal.inria.fr/hal-00932833}
\showURL{%
\tempurl}


\bibitem[Szabo(1996)]%
        {szabo96smart}
\bibfield{author}{\bibinfo{person}{Nick Szabo}.}
  \bibinfo{year}{1996}\natexlab{}.
\newblock \showarticletitle{Smart Contracts: Building Blocks for Digital
  Markets}.
\newblock \bibinfo{journal}{\emph{Extropy}}  \bibinfo{volume}{16}
  (\bibinfo{year}{1996}).
\newblock


\bibitem[{The ZeroMQ authors}(2021)]%
        {zeromq}
\bibfield{author}{\bibinfo{person}{{The ZeroMQ authors}}.}
  \bibinfo{year}{2021}\natexlab{}.
\newblock \bibinfo{title}{ZeroMQ}.
\newblock
\newblock
\urldef\tempurl%
\url{https://zeromq.org}
\showURL{%
\tempurl}
\newblock
\shownote{https://zeromq.org}.


\bibitem[Torres et~al\mbox{.}(2021)]%
        {Ferreira2021Frontrunner}
\bibfield{author}{\bibinfo{person}{Christof~Ferreira Torres},
  \bibinfo{person}{Ramiro Camino}, {and} \bibinfo{person}{Radu State}.}
  \bibinfo{year}{2021}\natexlab{}.
\newblock \showarticletitle{Frontrunner Jones and the Raiders of the {D}ark
  {F}orest: An Empirical Study of Frontrunning on the {Ethereum} Blockchain}.
  In \bibinfo{booktitle}{\emph{Proc of USENIX Sec.'21}}.
  \bibinfo{pages}{1343--1359}.
\newblock
\showISBNx{978-1-939133-24-3}
\urldef\tempurl%
\url{https://www.usenix.org/conference/usenixsecurity21/presentation/torres}
\showURL{%
\tempurl}


\bibitem[Tyagi and Kathuria(2021)]%
        {Tyagi@BlockchainScalabilitySol}
\bibfield{author}{\bibinfo{person}{Shobha Tyagi} {and}
  \bibinfo{person}{Madhumita Kathuria}.} \bibinfo{year}{2021}\natexlab{}.
\newblock \bibinfo{booktitle}{\emph{Study on Blockchain Scalability
  Solutions}}.
\newblock \bibinfo{publisher}{ACM}, \bibinfo{pages}{394–401}.
\newblock
\showISBNx{9781450389204}
\urldef\tempurl%
\url{https://doi.org/10.1145/3474124.3474184}
\showURL{%
\tempurl}


\bibitem[Wang and Kim(2019)]%
        {Wang2019FastChain}
\bibfield{author}{\bibinfo{person}{Ke Wang} {and} \bibinfo{person}{Hyong~S.
  Kim}.} \bibinfo{year}{2019}\natexlab{}.
\newblock \showarticletitle{FastChain: Scaling Blockchain System with Informed
  Neighbor Selection}. In \bibinfo{booktitle}{\emph{Proc. of IEEE
  Blockchain'19}}. \bibinfo{pages}{376--383}.
\newblock
\urldef\tempurl%
\url{https://doi.org/10.1109/Blockchain.2019.00058}
\showDOI{\tempurl}


\bibitem[Wood(2014)]%
        {wood2014ethereum}
\bibfield{author}{\bibinfo{person}{Gavin Wood}.}
  \bibinfo{year}{2014}\natexlab{}.
\newblock \showarticletitle{Ethereum: A secure decentralised generalised
  transaction ledger}.
\newblock \bibinfo{journal}{\emph{Ethereum project yellow paper}}
  \bibinfo{volume}{151} (\bibinfo{year}{2014}), \bibinfo{pages}{1--32}.
\newblock


\bibitem[Wood(2016)]%
        {wood2016polkadot}
\bibfield{author}{\bibinfo{person}{Gavin Wood}.}
  \bibinfo{year}{2016}\natexlab{}.
\newblock \showarticletitle{Polkadot: Vision for a heterogeneous multi-chain
  framework}.
\newblock \bibinfo{journal}{\emph{White Paper}}  \bibinfo{volume}{21}
  (\bibinfo{year}{2016}).
\newblock


\bibitem[Xu et~al\mbox{.}(2021)]%
        {Xu2021SlimChain}
\bibfield{author}{\bibinfo{person}{Cheng Xu}, \bibinfo{person}{Ce Zhang},
  \bibinfo{person}{Jianliang Xu}, {and} \bibinfo{person}{Jian Pei}.}
  \bibinfo{year}{2021}\natexlab{}.
\newblock \showarticletitle{SlimChain: Scaling Blockchain Transactions through
  off-Chain Storage and Parallel Processing}.
\newblock \bibinfo{journal}{\emph{Proc. VLDB Endow.}} \bibinfo{volume}{14},
  \bibinfo{number}{11} (\bibinfo{date}{jul} \bibinfo{year}{2021}),
  \bibinfo{pages}{2314–2326}.
\newblock
\showISSN{2150-8097}
\urldef\tempurl%
\url{https://doi.org/10.14778/3476249.3476283}
\showDOI{\tempurl}


\end{thebibliography}

\vfill\pagebreak

\appendix

\end{document}